\documentclass{article}

\def\arxiv{1}

\title{%
Computing Flows in Subquadratic Space}

\usepackage[bookmarks,colorlinks,breaklinks]{hyperref}
\hypersetup{urlcolor=blue, colorlinks=true, citecolor=green!50!black, linkcolor=blue}
\usepackage{multirow}
\usepackage[T1]{fontenc}
\usepackage{amsmath}
\usepackage{amsfonts}
\usepackage{amssymb}
\usepackage{amsthm}
\usepackage{thmtools}
\usepackage{pgffor}
\usepackage{xcolor}
\usepackage{cleveref}
\usepackage{graphicx}
\usepackage{geometry}
\usepackage{enumitem}
\usepackage{url}
\usepackage{tikz}
\usepackage{algorithm}
\usepackage{algpseudocode}
\usepackage{thm-restate}
\usepackage{circuitikz}
\usepackage{tcolorbox}

\declaretheorem[numberwithin=section,refname={Theorem,Theorems},Refname={Theorem,Theorems}]{theorem}
\declaretheorem[numberlike=theorem]{lemma}

\declaretheorem[numberlike=theorem]{definition}
\declaretheorem[numberlike=theorem]{claim}
\declaretheorem[numberlike=theorem]{fact}

\newcommand{\R}{\mathbb{R}}

\newcommand{\N}{\mathbb{N}}
\newcommand{\Z}{\mathbb{Z}}

\renewcommand{\tilde}{\widetilde}

\newcommand{\tO}{\tilde{O}}

\DeclareMathOperator*{\argmin}{argmin}
\DeclareMathOperator*{\argmax}{argmax}

\newcommand{\poly}{\operatorname{poly}}
\newcommand{\polylog}{\operatorname{polylog}}

\newcommand{\diag}{\operatorname{diag}}

\DeclareMathOperator{\round}{round}

\newcommand{\unit}{\chi}

\foreach \x in {A,...,Z}{%
	\expandafter\xdef\csname m\x\endcsname{\noexpand\mathbf{\x}}
}
\foreach \x in {A,...,Z}{%
	\expandafter\xdef\csname om\x\endcsname{\noexpand\overline{\noexpand\mathbf{\x}}}
}

\foreach \x in {A,...,Z}{%
	\expandafter\xdef\csname c\x\endcsname{\noexpand\mathcal{\x}}
}

\renewcommand{\o}[1]{\noexpand{\overline{#1}}}
\newcommand{\otau}{\noexpand{\overline{\tau}}}
\newcommand{\omu}{\noexpand{\overline{\mu}}}
\newcommand{\osigma}{\noexpand{\overline{\sigma}}}
\newcommand{\os}{\noexpand{\overline{s}}}
\newcommand{\ot}{\noexpand{\overline{t}}}
\newcommand{\og}{\noexpand{\overline{g}}}
\newcommand{\ou}{\noexpand{\overline{u}}}
\newcommand{\ov}{\noexpand{\overline{v}}}
\newcommand{\ow}{\noexpand{\overline{w}}}
\newcommand{\ox}{\noexpand{\overline{x}}}
\newcommand{\oy}{\noexpand{\overline{y}}}
\newcommand{\oz}{\noexpand{\overline{z}}}
\newcommand{\oc}{\noexpand{\overline{c}}}
\newcommand{\oh}{\noexpand{\overline{h}}}
\newcommand{\of}{\noexpand{\overline{f}}}
\newcommand{\oT}{\noexpand{\overline{\mT}}}

\newcommand{\init}{\mathrm{(init)}}
\newcommand{\target}{\mathrm{(target)}}
\newcommand{\final}{{(\mathrm{final})}}

\newcommand{\OPT}{\mathrm{OPT}}

\makeatletter
\renewcommand{\paragraph}{%
	\@startsection{paragraph}{4}%
	{\z@}{1.25ex \@plus 1ex \@minus .2ex}{-1em}%
	{\normalfont\normalsize\bfseries}%
}
\makeatother

\makeatletter
\renewcommand{\subparagraph}{%
	\@startsection{paragraph}{4}%
	{\z@}{1.25ex \@plus 1ex \@minus .2ex}{-1em}%
	{\normalfont\normalsize\bfseries}%
}
\makeatother

\author{Jan van den Brand\footnote{Georgia Institute of Technology, \texttt{vdbrand@gatech.edu}} \and Zhao Song\footnote{\texttt{magic.linuxkde@gmail.com}} \and Albert Weng\footnote{Georgia Institute of Technology, \texttt{albweng@gatech.edu}}}
\date{}

\begin{document}
\pagenumbering{roman}
\maketitle
\begin{abstract}

Space complexity is a critical factor in various computational models, including streaming, parallel/distributed computing, and communication complexity.
We study the space complexity of the minimum-cost flow problem, a generalization of the st-max flow problem, focusing on computing flows in subquadratic space. 
In the general case with arbitrary capacities, minimum cost and $st$-maximum flows can use up to $\Omega(n^2)$ edges, so computing the flow on each edge (rather than just the size/cost)
seems impossible in subquadratic space. 
Indeed, there are lower bounds proving quadratic space is needed to store the flow on every edge, which has been used to prove lower bounds on streaming algorithms.
However, we show that these lower bounds can be circumvented, opening up improvements for streaming and communication complexity.

For a directed graph with integer capacities and costs bounded by $W$, we provide a $\tilde O(n^{1.5}\log (W/\epsilon))$-space $\tilde O(\sqrt{n} \log(W/\epsilon))$-pass streaming algorithm,
which during the last pass
returns the flow on each edge
up to an additive error of $\epsilon$.
Crucially, the algorithm does not return the flow at the end of the last pass but returns the flow on an edge, as the edge is read in the stream.
This allows us to circumvent existing $\Omega(n^2)$ space lower bounds.
In the 2-party communication model, our algorithm implies $\tilde O(n^{1.5}\log^2 W)$ bits of communication.

\end{abstract}

\clearpage
\tableofcontents

\clearpage
\pagenumbering{arabic}
\section{Introduction}

Space complexity is a important factor in streaming, parallel/distributed computing, communication complexity, and data structures, requiring computation and representation of objects using limited space.
In this work, we study the space complexity of the minimum-cost flow problem, a generalization of the $st$-max flow problem. 
We focus on computing such flows in subquadratic space
in the multi-pass streaming model, contrasting quadratic space lower bounds from \cite{ChakrabartiJWY25,AssadiCK19,AssadiGLMM24}.
In particular, our goal is to compute the actual flow on each edge, rather than just the cost of the min-cost flow or the size of the maximum flow.

In the minimum-cost flow problem, we are given a directed $n$-node graph $G=(V,E)$, a vertex-demand vector $d\in\Z^V$, edge capacities $u\in\Z^E_{\ge0}$ and edge costs $c\in\Z^E$. The minimum-cost flow is a flow $f\in\R_{\ge 0}^E$ that satisfies the vertex-demands and edge capacities, while minimizing the cost $\sum_{e\in E} c_e\cdot f_e$.
While the min-cut $W\subset V$ (and thus size of a max flow) can be stored in $O(n)$ space, observe that flows in general can use $\Omega(n^2)$ edges. Thus it is not clear if one can compute the flow $f\in\R_{\ge0}^E$ (i.e., the amount of flow on each edge) in subquadratic space.

\subparagraph{Streaming Algorithms}

In the streaming model, the input (i.e., edge set) is given to the algorithm one-by-one as a stream of data. The algorithm is only allowed a small amount of space and thus cannot store the entire input. In a multi-pass algorithm, one may take multiple passes over the input stream.
This computational model is motivated by large amounts of data being faster to read sequentially than via random access, e.g., when the data is too large for working memory and is read directly from hard-drive.
Bipartite matching (a simpler special case of max flow) is one of the most successfully studied problems in streaming literature, both in the approximate 
(see, e.g.,\cite{EggertKS09,AhnG11,EggertKMS12,GoelKK12,KonradMM12,Kapralov13,AhnG18,BeckerFKL21,AssadiJJST22,AssadiS23,Assadi24} and references therein)
and exact setting \cite{lsz+23,AssadiJJST22}. However, for flows much less is known.
For undirected maximum flows, one can $(1+\epsilon)$-approximate the size of the max flow in $\tilde O(n)$ space via duality 
by constructing a cut sparsifier (see, e.g., \cite{ag09,AhnGM12,KyngPPS17,ckst19,ChakrabartiJWY25}).
\cite{RubinsteinSW18} show one can compute the exact $st$-min cut (i.e., size of max flow) on unweighted undirected graphs in $\tilde O(n^{5/3})$ space.
Computing the exact flow size/cost on weighted graphs in subquadratic space remains open.

When it comes to computing not just the size/cost of the flow but the flow $f\in\R^E_{\ge0}$ itself, existing upper bounds are substantially weaker.
While sparsification approaches allow a $(1+\epsilon)$-approximation of the size, only crude $\poly(n)$ approximations are known for computing flow $f$ in subquadratic space \cite{ChakrabartiJWY25}.
Indeed, this gap between size and flow $f$ can be explained and \cite{ChakrabartiJWY25} provide respective lower bounds.

For intuition, why representing (i.e., not just computing) flows in subquadratic space is impossible, consider 
a complete bipartite graph with edges oriented left to right, an extra vertex $s$ connected to the left vertices, and an extra vertex $t$ reachable from all right vertices.
The edges from $s$ and edges to $t$ have very high capacity, so the max flow has $f_e = u_e$ for each edge $e$ of the bipartite graph.
Thus storing/representing all (or even most, in case of approximation) $f_e$ in subquadratic space means we can store all (or most) capacities $u_e$ in subquadratic space, but that is $\Theta(n^2)$ information ($u \in \R^E$ could be any $n^2$-sized binary bit string).
This lower bound argument was formalized by \cite{ChakrabartiJWY25}\footnote{\cite{ChakrabartiJWY25} stated this as $\ell_\infty$-regression, but the statement is equivalent to max-flow. 
\ifdefined\arxiv
See \Cref{sec:lower_bound}.
\else 
Details in the full version.
\fi}, who proved that for every encoding algorithm $\cE(G)\in \{0,1\}^d$ (computing a max-flow on $G$ and then encoding it in $d$ bit) and decoding algorithm $\cA$ (where $\cA(\cE(G))$ returns a $(1-1/36)$-approximation of the max-flow) needs $d = \Omega(n^2)$. %
Thus it's impossible to encode max/min-cost flows $f$ in subquadratic space.
Importantly, the lower bound by \cite{ChakrabartiJWY25} applies not just to streaming but to any encoding scheme. Thus the argument can also be used to argue about the communication complexity for communicating $f$.

\subparagraph{Communication Complexity}
In the 2-party (or multi-party) communication model, the edges of the input graph are distributed among multiple parties.
The task is to solve some graph problem while using as little communication as possible.
This area, too, had a lot of success for various graph problems \cite{IvanyosKLSW12,BlikstadJMY25,FlinM24,DobzinskiNO14,PapadimitriouS82,BabaiFS86,HajnalMT88,DurisP89,MandePSS24}, but flows remain elusive.
For related problems like bipartite matching or transshipment, an $\tilde O(n)$ upper bound is known \cite{BlikstadBEMN22,IvanyosKLSW12}, crucially relying on the fact that the optimal solution consists of at most $O(n)$ edges.
For unweighted undirected $st$-min cut, a $\tilde O(n^{5/3})$ communication protocol is known \cite{RubinsteinSW18}, later improved to $\tilde O(n^{11/7})$ \cite{JiangNS26}. We remark that on unweighted graphs the max-flow can use at most $O(n^{1.5})$ edges\footnote{%
\ifdefined\arxiv
See \Cref{thm:unweightedflow}
\else 
See full version for details.
\fi}, 
so the problem is substantially easier in the unweighted setting.
A subquadratic communication algorithm for weighted flows (but also the simpler weighted $st$-min-cut problem where solutions have size $O(n)$) was stated as open problem in \cite{BlikstadBEMN22}.

When it comes to flow $f$, rather than its cost/size, note that not just the \emph{computation} is difficult, but just \emph{communication} of the solution is already a challenge.
Namely, assume one party already knows all the edges and the flow. How would that party communicate the flow to the other party?
This is clearly a simpler problem than jointly computing the flow since the 1st party could just throw away the extra information.
However, how would one communicate the flow if no low space representation of the flow is possible? 
In particular, the communicated messages would be an encoding of the flow.

\subparagraph{Computing Flow In Subquadratic Space} 
We show that computation of the min-cost/max flows in subquadratic space is possible.
In addition to the exact flow cost/size, our algorithm also returns the flow $f_e$ on each edge $e\in E$.
This is unexpected, as it appears to contradict the lower bound of \cite{ChakrabartiJWY25}.
We can circumvent the lower bound, because it assumes a decoding algorithm of form $\cD(M)$ for encoding $M\in\{0,1\}^*$ of the flow, i.e., decoding must be performed without access to the graph.
The lower bound does not rule out a decoding algorithm of form $\cD(M, e)$ that receives edge $e$ as part of the input in addition to some subquadratic space encoding $M\in\{0,1\}^*$ of the flow.
The assumption that we have access to the edge during decoding is naturally satisfied in many low space models.
In the communication model, each party knows their own (possibly $\Omega(n^2)$ sized) set of edges.
For multi-pass streaming algorithms, the edges are available via another pass over the input. In particular, during the last pass over the edges, the algorithm returns the flow on each edge when it is read in the stream.
This approach allows us to compute flows in subquadratic space, even though the flow uses up to $\Omega(n^2)$ many edges.

Our main technical contribution is such an encoding/decoding scheme adapted to the framework of ``robust interior point methods'' for linear programming. In the context of streaming algorithms, the mere existence of such encoding does not suffice, but we show that the computation of this encoding can be done in low space, too. %
Using the encoding with the robust interior point method of \cite{bll+21} (which solves linear programs/min-cost flow in $\tilde O(\sqrt{n})$ iterations), we present a $\tilde O(\sqrt n)$-pass, $\tilde O ( n^{1.5})$-space streaming algorithm, which can be implemented with $\tilde O(n^{1.5})$ communication in the 2-party model.
\subsection{Our Results}

We present the following result for computing min-cost flows in the multi-pass streaming model.

\begin{restatable}[Min-cost Flow, Streaming]{theorem}{mainStreamFlows}\label{thm:main_stream}
    There is a randomized multi-pass streaming algorithm that, given an input graph $G=(V,E)$ with integer capacities $u\in[1,W]^E$ and costs $c\in[1,W]^E$, vertex demand vector $b\in\R^V$ and accuracy parameter $\epsilon\in[0,1]$, computes w.h.p.~a min-cost flow in $\tilde{O}(n^{1.5}\log(W/\epsilon))$ space and $\tilde{O}(\sqrt{n}\log(W/\epsilon))$ passes over the input.
    The total time is $\tilde O(mn \log^2 (W/\epsilon))$.

    The algorithm returns a randomized $\tilde O(n^{1.5}\log (W/\epsilon))$-space data structure that supports edge-queries: For any given edge $e=((w,v),c_e,u_e)$, the query returns in $\tilde O(\sqrt{n} \log (W/\epsilon))$ time an estimate $\of_e$ of the flow on that edge, with $0 \le \of_e \le u_e$ and %
    $|\of_e - f_e| \le \epsilon$ with high probability.
    
    Alternatively, we can also assume that during the last pass, the algorithm returns for each given edge $e=((w,v),c_e,u_e)$ the flow $\of_e$ on that edge.
\end{restatable}
While the streaming area often focuses on small number of passes (e.g., constant or polylog), it was proven that polynomial number of passes may be necessary for flows in the weighted case.
In the presence of capacities as large as sub-exponential in $n$, \cite{AssadiCK19} showed that a streaming algorithm for computing the size of the max-flow in $p$ passes needs $\Omega(n^2/p^5)$ space. This lower bound was later improved to $\Omega(n^2/p^3)$ space \cite{AssadiGLMM24}. 
For smaller polynomially-sized capacities, \Cref{thm:main_stream} is the first upper bound for computing the exact size of the max-flow in the weighted case (when interested in the size, we can simply round the result to nearest integer). It is also the first result to return the flow $f \in \R^E$ with high accuracy rather than the cut or cost/size of the flow.

Previous work was either highly inaccurate with polynomial error when returning the flow \cite{ChakrabartiJWY25}, or when exact, previous work was only for the unweighted undirected case \cite{RubinsteinSW18} and returned only the cut but not the flow. 
Observe that for unweighted graphs, any max flow uses at most $O(n^{1.5})$ edges
\ifdefined\arxiv
(\Cref{thm:unweightedflow}),
\else 
(see full version for details),
\fi
but for weighted graphs the number of used edges can be as large as $O(n^2)$. 
So minimizing space complexity is generally harder for the weighted case.

\subparagraph{Regarding Accuracy}
If we care only about the cost of the flow (or size of max-flow), then the algorithm's output is exact by choosing $\epsilon \le 1/3$ and then rounding the computed cost to the nearest integer.

If the minimum-cost flow is unique, then 
rounding each $f_e$ to the nearest integer also results in the exact flow on each edge. 
If the flow is not unique, then the algorithm might compute a fractional flow and thus rounding is not guaranteed to provide an exact solution for the edges.
However, given the logarithmic complexity dependence on $\epsilon$, we can reduce the error to an arbitrary small $\epsilon = 1/\poly(n,W)$ on each edge. Thus we get an almost exact fractional min-cost flow $f\in\R^E_{\ge0}$.

Usually, uniqueness can be guaranteed via Isolation Lemma \cite{ChariRS95,DaitchS08}, which adds small random perturbation to the problem instance, but this requires large additional space to store $\Omega(n^2)$ random bits. 
It is open whether $o(n^2)$ random bits suffice to isolate maximum/min-cost flows, and in our streaming algorithm we cannot afford to store these.
This issue does not occur in the setting of communication complexity, as each player can independently store their own perturbed edge weights.

\subparagraph{Communication Complexity}

There is a general reduction from 2-party communication to streaming, resulting in $O(space\times passes)$ amount of communication, by simulating the streaming algorithm and sending the entire memory in each pass.
While our streaming algorithm needs $\tilde O(n^{1.5})$ total space, it generates only $\tilde O(n)$ new information in each of the $\tilde O(\sqrt{n})$ passes. Thus it suffices to send only the $\tilde O(n)$ additional information in each pass, leading to $\tilde O(n^{1.5})$ communication.

\begin{restatable}[Min-cost Flow, Communication]{theorem}{mainCommunication}\label{thm:main_communication}
    There is a randomized communication protocol that, given an input graph $G=(V,E)$ with integer capacities $u\in[1,W]^E$ and costs $c\in[1,W]^E$, vertex demand vector $b\in\R^V$, where the edges are split among two parties,
    computes w.h.p.~an exact min-cost flow in $\tilde{O}(n^{1.5}\log(W))$ communication.
    At the end, each party knows the amount of flow on each of their own edges.
\end{restatable}
Unlike \Cref{thm:main_stream}, there is no $\epsilon$ dependence in \Cref{thm:main_communication}, because we can make the flow unique via Isolation-Lemma.
A subquadratic communication algorithm for weighted flows (but also the simpler weighted $st$-min-cut problem where solutions have size $O(n)$) was stated as open problem in \cite{BlikstadBEMN22}.

The same $\tilde O(n^{1.5})$-communication bound %
is also obtained in concurrent work \cite{GholizadehJ25}. We remark that 2-party communication is an easier model than multi-pass streaming, because each party can use unlimited amount of space. In particular, \cite{GholizadehJ25} does not give a streaming algorithm because each party explicitly stores $f_e$ on their own edges $e$, thus using $\Omega(n^2)$ space.

\subparagraph{Space and Bit-Complexity}
In above theorems, space is measured in words where each word can store a real number.
The algorithms also work over Word-RAM and finite bit numbers, where each number is stored as fixed-point number using $\tilde O(\log(nW/\epsilon))$-bit.
Thus if we want to measure space complexity of \Cref{thm:main_stream,thm:main_communication} in bits, all space, communication and time complexities increase by a logarithmic $\tilde O(\log(nW/\epsilon))$ factor.

\subsection{Techniques}

Our algorithm is based on solving the linear program representation of min-cost flow.
The linear program representation of related problems such as bipartite matching or transshipment to obtain streaming algorithms has been used before, e.g., in \cite{AssadiJJST22,lsz+23,AhnG11}. 
A natural approach for matching is to maintain the smaller $n$-dimensional dual iterate (i.e., fractional vertex cover).
However, this does not extend to min-cost/maximum flows, because the capacity constraints increase the dimension of the dual solution from $n$ to $m+n$.

Our algorithm works by implementing the interior point method/central path method by \cite{bll+21}, while maintaining the primal solution (i.e., flow) in low space, rather than the dual.
One can show that for flows on the central path, an $\tilde O(n)$ space representation is possible\footnote{Technical details: Centered flows are minimizers of some potential function $\Phi(x)$ subject to $\mA^\top x = b$ for edge-vertex incidence matrix $\mA$ and demand $b\in\R^V$. Minimizers satisfy $\nabla \Phi(x) = \mA y$ for some $y\in\R^V$, thus by storing $y$ we can reconstruct $x$.}.
Unfortunately, this observation does not directly imply a streaming algorithm because in each iteration, the central path method constructs intermediate flows that are not sufficiently centered to be stored in $\tilde O(n)$ space. 
This issue is exacerbated when using the fast $\tilde O(\sqrt{n})$ iteration central path of the Lee-Sidford-barrier \cite{ls14} as used in \cite{bll+21}. Their definition of centrality uses Lewis-weights whose efficient computation uses sketching. 
Due to the resulting approximation error, one cannot compute sufficiently centered flows to store them in $\tilde O(n)$ space.
Thus we develop a method that can store non-centered flows in small space.

We do this by storing the steps of each iteration rather than the flow $x$.
In the context of flows, each step of the central path method is equivalent to augmenting the current flow $x$ by an electric circulation.
Rather than storing flow $x$ directly, we store the sequence of $\tilde O(\sqrt{n})$ augmenting electric flows.
Each circulation is an $m$-dimensional vector (i.e., every edge carries some flow), but we show that each circulation can be stored in only $\tilde O (n)$ space. %
Thus we obtain a $\tilde O(n^{1.5})$ space representation of the min-cost flow.

\noindent
The main technical ideas/ingredients for storing electric circulations are as follows.

(i) The $\tilde O(\sqrt n)$ iterations framework \cite{ls14,bll+21} requires computation of Lewis-weights, a generalization of leverage scores and effective resistances. Via sketching, we can store the Lewis-weights in only $\tilde O(n)$ space. Given any of the $O(n^2)$ edges, we can then query/compute the (approximate) Lewis-weight from only the $\tilde O(n)$ space representation. 

(ii) While electric flows can be stored via vertex potentials (i.e., in $\tilde O(n)$ space), electric circulations generally need $\Omega(m)$ space. That is because electric circulations are a combination/difference $g-f$ of an arbitrary flow $g\in\R^E$ and an electric flow $f\in\R^E$, where $f$ and $g$ satisfy the same vertex demands. 
To store the electric circulation $g-f$ in low space, we must store $g$ in low space.
Via the robust interior point framework \cite{bll+21}, one can reduce the dimension of $g$ from $m$ to $\tilde O(1)$, and can thus be stored in small space. 
The main barrier is showing that this projection $E\to \{1,....,\tilde O(1)\}$ can be computed from the stored information.

(iii) Electric flows can only be stored via vertex potentials if we know the edges' resistances. Retrieving the flow on edge $e=(u,v)$ from vertex potentials $y\in\R^V$ (i.e., $f_e = (y_v-y_u)/r_e$) requires knowing the edge's resistance $r_e$.
In \cite{ls14,bll+21}, the resistance $r_e$ is defined w.r.t.~the edge's Lewis-weights, and the current congestion of that edge. This allows for an inductive argument:
Given an edge $e$ and its capacity $c_e$, assuming we know the amount of flow on edge $e$ after $t$ electric circulations, then we can compute the resistance $r_e$ of that edge from the current congestion and the stored Lewis-weight. From the resistance, we can then obtain how much flow the next electric circulation will route through $e$. This then yields the flow on edge $e$ after $t+1$ electric circulations.
After $\tilde O(\sqrt{n})$ iterations of this argument, we know the final amount of flow on that edge, i.e., how much flow is routed on $e$ by the min-cost flow.

\subparagraph{Remark on General Linear Programs}
Given that our work is based on interior point methods for general $n\times m$-dimensional ($m\ge n$) linear programs, it is a natural question whether our algorithm extends to other applications such as $\ell_1$- or $\ell_\infty$-regression, MDPs, etc.
Internally, our algorithm stores $\tilde O(\sqrt{n})$ many $n$-dimensional vectors, and $\tilde O(n)$ constraints of the linear program. So for a general linear program with $k$ non-zeros per columns ($k=2$ in case of flow), it would store $\tilde O(n^{1.5}+kn)$ numbers. However, this assumes Real-RAM model, i.e., when ignoring the bit-complexity.
When considering the actual number of bits (e.g., Word-RAM), the linear systems to be solved in each iteration pose additional technical challenges that do not occur when restricting to max/min-cost flow where the linear systems are simple Laplacians. 
Hence, this work focuses on the flow applications, where the (sparse) linear systems can be solved in low space in Word-RAM via Laplacian solvers.

\subsection{Other Related Work}

Before the semi-streaming lower bounds of $\Omega(n^2/p^{5})$ \cite{AssadiCK19} and $\Omega(n^2/p^3)$ \cite{AssadiGLMM24} for $p$ passes, there have been $n^{1+\Omega(1/p)}$ space lower bounds \cite{ChakrabartiG0V20,GuruswamiO16}.
For small number $o(\sqrt{\log n})$ passes, there is an $n^{2-o(1)}$ space lower bound for $st$-reachability, which extends to directed flows \cite{ChenKPSSY21}.

While our focus is on $st$-flows and thus $st$-min cuts, there has been work on global min cuts in the streaming model.
There are $\tilde O(1)$ pass $\tilde O(n)$ space algorithms for global min cuts \cite{MukhopadhyayN20,AssadiD21,RubinsteinSW18}.
These also transfer to the 2-party communication setting, where $\Theta(n^{1.5})$ upper and lower bound are known for minimum vertex cut \cite{BlikstadJMY25}. 

Our algorithm is based on solving linear programs in the streaming model. Other work using this approach (though not for flows, but instead for other linear programs such as bipartite matching) include \cite{lsz+23,AhnG11,AssadiKZ19,ChanC05}.
From the communication perspective, linear programs have been studied in \cite{VempalaWW20,GhadiriLPSWY24}.
The crucial aspect is that previous work studies linear programs of form $\min c^\top$ subject to $\mA^\top x = b, x\ge 0$ for tall matrices $\mA$, but max-flow requires additional capacity constraints $x \le u$.
Encoding those within the matrix $\mA$ turns the matrix from tall into an almost square matrix, thus losing any space/communication gains from the small width of $\mA$.

\section{Preliminaries}
\label{sec:preli}

We write ``with high probability'' (w.h.p.) if the failure probability is at most $1/\poly(n)$.
We use $\unit_e$ to denote the $e$-th standard unit vector.
For an integer $n$, we use $[n]$ to denote the set $\{1,2,\cdots,n\}$. We use $\tO(\cdot)$ to hide $\poly(\log m)$ factors. 
In addition to $O(\cdot )$ notation, for two functions $f,g$, we use the shorthand $f \lesssim g$ (resp. $\gtrsim$) to indicate that $f \leq C \cdot g$ (resp. $\geq$) for an absolute constant $C$. 
For a vector $v \in \R^m$, we use $\| v \|_{\infty}$ to denote its $\ell_{\infty}$ norm, i.e., $\| v \|_{\infty}: \max_{i \in [m]} |v_i|$.
We use $\| v \|_p$ to denote its $\ell_p$ norm, i.e., $\| v \|_p:= (\sum_{i=1}^m |v_i|^p)^{1/p}$.
For a positive semi-definite matrix $\mM$ we write $\|v\|_\mM := (v^\top \mM v)^{1/2}$.
For vectors $x,s,\tau\in\R^m$ we write $\mX,\mS,\mT$ for the diagonal matrices with $\mX_{i,i}=x_i, \mS_{i,i}= s_i, \mT_{i,i}=\tau_i$ on the diagonals for $i\in[m]$.
Similarly, for function $\phi:\R\to\R$, we write $\Phi(x)$ for the diagonal matrix with $\Phi(x)_{i,i}=\phi(x_i)$ on the diagonals for $i\in[m]$.
Given vectors $a, b \in \R^m$, we use $ab$ to denote entry-wise products, i.e., $(ab)_i = a_ib_i$.
We define $\cosh(x) := \frac{1}{2}(\exp(x) + \exp(-x))$. We define $\sinh(x):= \frac{1}{2} ( \exp(x) - \exp(-x) )$. Note that $\cosh(x)' = \sinh(x)$ and $\sinh(x)' = \cosh(x)$.

The \emph{leverage scores} $\sigma(\mA) \in \R^m_{>0}$ of a matrix $\mA\in\R^{m\times n}$ ($m \ge n$) are defined as $\sigma(\mA)_i := (\mA (\mA^\top \mA)^{-1} \mA^\top)_{i,i}$ for $i \in [m]$.

For $a,b\in\R$ and $\epsilon>0$ we write $a \approx_\epsilon b$ when $\exp(-\epsilon) a \le b \le \exp(\epsilon) a$.
We extend the notation to vectors when the approximation holds entry-wise.
For PSD matrices $\mA,\mB\in\R^{n\times n}$ we write $\mA \approx_\epsilon \mB$ when for all vectors $v\in\R^n$ $v^\top \mA v \approx_\epsilon v^\top \mB v$, i.e., $\mA,\mB$ are spectral approximations of each other.

    We say a vector $v\in\R^E$ is given in implicit representation with space $O(S)$ and query time $O(T)$, if we have a data structure that uses $O(S)$ space and supports queries that receive as input edge, edge-cost, and edge-capacity $(e,c_e, u_e)$ and then return $v_e$ in $O(T)$ time.

The line of work \cite{ag09,kl11,AhnGM12,KyngPPS17,klm+17,kmm+19,knst19,ckst19,lsz+23} developed streaming algorithms for spectral sparsifiers.

\begin{lemma}[{\cite[Theorem 1.1]{KyngPPS17}}]\label{lem:spectralsparsifier}
    Let $\mA \in \R^{E\times V}$ be an incidence matrix, $w\in\R^E$ be some weights where any $w_e$ can be queried in time $T_w$, and let $\epsilon > 0$.
    Then within a single pass over the rows of $\mA$ (i.e., edges), we can compute an $\tilde O(n/\epsilon^2)$-sized spectral sparsifier $\mH \approx_\epsilon \mA^\top \mA$. This takes $\tilde O(mT_w)$ total time.
\end{lemma}

This lemma was proven for weighted incidence matrices in \cite{KyngPPS17} where $w_e$ is given together with edge $e$ in the stream.
It directly implies the above, by instead querying $w_e$ when edge $e$ is read from the stream.

\section{Technique Overview}
\label{sec:overview}

Our algorithm and data structure are based on interior point methods for linear programs.
In particular, we show that the algorithm of \cite{bll+21} can be implemented in $\tilde O(n^{1.5})$ space.
The main barrier for this is storing the $m$-dimensional flow vector in subquadratic space.
Given that a lot of work on streaming and communication complexity for flows and matching is combinatorial in nature \cite{%
AssadiCK19,AssadiD21,RubinsteinSW18,BlikstadJMY25,IvanyosKLSW12,AssadiKL17,CrouchS14,McGregor05,FeigenbaumKMSZ05}, we start with a more graph-oriented perspective for readers with that background.
\Cref{sec:graphperspective} sketches how we store our flow $x$ in $\tilde O(n^{1.5})$ space, without going too far into the details about the interior point method.

The next subsections \ref{sec:overview:reviewipm} and \ref{sec:overview:implement} then give more details of our algorithms using the optimization perspective of interior point methods and central paths.
In order to outline time and space complexity guarantees of our %
streaming result \Cref{thm:main_stream}, we need more details about these optimization methods. 
After outlining these methods in \Cref{sec:overview:reviewipm}, we explain in \Cref{sec:overview:implement} how to implement this interior point methods in the streaming model, resulting in \Cref{thm:main_stream}.

\subsection{High-Level Idea: Storing Flows in \texorpdfstring{$\tilde O(n^{1.5})$}{} Space}
\label{sec:graphperspective}
\label{sec:overview:graph}
Here we outline how to store a flow $x\in\R^E$ in low-space, such that when given any edge $e\in E$, its cost $c_e$, and capacity $u_e$, then we can query/compute $x_e$.

Consider the linear programming definition of minimum cost flow. For an $E\times V$ incidence matrix $\mA\in\{-1,0,1\}^{E\times V}$, edge-cost vector $c\in\R^E$, demand vector $b\in\R^V$, and edge capacities $u\in\R^E$, the following linear program models minimum-cost flows.
\[
    \min_{x\in\R^E} c^\top x \text{ subject to }~%
    \mA^\top x = b~~~%
    0 \le x \le u
\]
A common idea for computing minimum cost flows is to start with some initial flow, then repeatedly augment the flow via negative cost cycles or circulations. When solving minimum cost flows via interior point methods such as \cite{m13,ls14,m16,amv20}, these circulations are so called ``electric circulations.'' 
That is, for some feasible flow $x \in \R^E$, we augment it by a circulation $\delta_x \in \R^E$ where
\begin{align*}
    x \gets&~ x + \delta_x\\
    \delta_x =&~ \left(\mI - \mR_x^{-1} \mA (\mA^\top \mR_x^{-1} \mA)^{\dagger} \mA^\top\right) g_x.
\end{align*}
Here $g_x\in\R^E$ is some flow and $\mR_x$ is an $E\times E$ diagonal matrix with $(\mR_x)_{e,e}$ being resistance of edge $e$. (Both $g_x$ and $\mR_x$ will be defined later to depend on $x$).
Flow $\delta_x$ is a circulation, because it satisfies 0-demand $\mA^\top \delta_x = 0$.
It is an electric circulation because
\begin{align}
f = -\mR_x^{-1} \mA (\mA^\top \mR_x^{-1} \mA)^{\dagger} \mA^\top g_x \label{eq:overview:f}
\end{align}
is an electric flow for demand $d=-\mA^\top g_x \in \R^V$ and resistances $\mR_x$ (i.e., among all demand $d$ flows, $f$ minimizes the energy $\sum_{e\in E} (\mR_x)_{e,e} f_e^2$).

The idea of our low space representation is that we only store some initial flow $x^0 \in\R_{\ge 0}^E$, and then store all the augmenting circulations $\delta_x^1,\delta_x^2,...$. 
We will argue that each $\delta_x^t \in \R^E$ can be stored in only $\tilde O(n)$ space. 
It was proven in \cite{ls14} that it takes $\tilde O (\sqrt{n} \log(W/\epsilon))$ iterations/augmentations to find a min-cost flow whose cost is at most an additive $\epsilon$ off from optimal. Thus $\tilde O(n^{1.5} \log (W/\epsilon))$ total space suffices\footnote{For the rest of the overview, we will hide the $\log(W/\epsilon)$ term in $\tilde O$-notation.}. We can then query the amount of flow on any edge $e\in E$ by computing $x^0_e + \sum_{t=1}^{\tilde O (\sqrt{n})} (\delta_x^t)_e$.

Crucial for our low space representation is that the $(t+1)$-th augmenting circulation $\delta_x^{t+1}$ depends on the $t$-th flow $x^t = x^0+\sum_{k=1}^t \delta_x^k$.
The fact that we can store circulation $\delta_x^{t+1}$ in only $\tilde O(n)$ space relies heavily on the following property: when querying/computing $(\delta_x^{t+1})_e$ for some edge $e$, we already know the flow $x^t_e$ on that edge.
This assumption is given by induction: if we know $x^{t-1}_e$, then we can compute $(\delta_x^t)_e$, thus we know $x^t_e = x^{t-1}_e + (\delta_x^t)_e$, and so forth.
This way we iteratively reconstruct $x^t_e$ and $(\delta_x^t)_e$ from $t=1,2,3,...$ all the way to the final/optimal flow on that edge $e$.

To outline these low-space representations, let us start with the initial flow $x^0$.

\ifdefined\arxiv
\subparagraph{Storing the initial flow $x^0$ (\Cref{sec:initial_point:x0}):}
\else 
\subparagraph{Storing the initial flow $x^0$:}
\fi
We would like to start with an initial flow $x_0$ that is feasible and can be stored in low space. Fortunately, standard constructions for the initial flow in non-streaming settings \cite{bll+21,ls14,GaoLP21,BrandGJLLPS22} can also be stored in low space and thus directly apply. Consider the following construction.

We pick $x^0_e = u_e/2$, i.e., each edge carries half its capacity as flow.
This flow $x^0$ does not yet satisfy the demand vector $b\in\R^V$, but we can fix this as follows:
add an extra star to the graph (i.e., one vertex that is connected to every other vertex) and route the missing/additional flow along the star edges.
The star edges are assigned very large $\poly(n\|c\|_\infty)$ cost so the optimal min-cost flow will not use them, i.e., addition of these edges to the graph does not change optimal solution.

We can store this additional star and the amount of flow $x^0_e$ on the star-edges in $\tilde O(n)$ space since there are only $n$ additional edges.
For all edges $e$ of the original graph, we can return $x^0_e = u_e/2$ since the capacity $u_e$ of the edge is given during the query.

\ifdefined\arxiv
\subparagraph{Storing the Electric Flow $f$ (\Cref{sec:ipm:primal,sec:ipm:dual}):}
\else 
\subparagraph{Storing the Electric Flow $f$:}
\fi
To argue why we can store each electric circulation $\delta_x$ in $\tilde O(n)$ space, we argue the storage of $g\in\R^E$ and the electric flow separately.

An electric flow $f$ (see \eqref{eq:overview:f}) is given by vertex potentials $\delta_y := (\mA^\top \mR_x^{-1} \mA)^{\dagger} \mA^\top g \in \R^V$ which can easily be stored in $\tilde O(n)$ space.
Then for any edge $e=(w,z)$, we can compute the electric flow $f_e$ on that edge via 
\[f_e = (\mR^{-1}_x\mA \delta_y)_e = ((\delta_y)_z - (\delta_y)_w)/(\mR_x)_{e,e}.\]
However, this also needs access to $(\mR_x)_{e,e}$.
To guarantee $\tilde O(\sqrt{n})$ iterations suffice, these resistances must be
\[(\mR_x)_{e,e}=\tau_e\cdot\left(\frac{1}{x_e^2}-\frac{1}{(u_e-x_e)^2}\right).\]
Here term $\tau\in\R^E$ are so called Lewis-weights, which can be interpreted as a measure of importance for each edge $e$, and they are a generalization of effective resistance. Without this $\tau$, the number of iterations would increase from $\tilde O(\sqrt{n})$ to $\tilde O (\sqrt{m}) = \tilde O(n)$, and we could no longer guarantee subquadratic space.

The term $1/x_e^2 -1/(u_e-x_e)^2$ is motivated by the resistance $(\mR_x)_{e,e}$ going towards $\infty$ when $x_e$ or $(u_e-x_e)$ are close to zero. 
When the resistance is high, then the electric flow will not route a lot of flow through that edge, which is needed to guarantee the capacity constraints $0\le x+\delta_x \le u$.

Observe that, given edge $e$ and its capacity $u_e$, we can compute $x_e$, so we can also compute $1/x_e^2 -1/(u_e-x_e)^2$.
If we can also compute $\tau_e$ and $g_e$, then we get $(\delta_x)_e$, and our iterative argument for computing $x^t_e=x^{t-1}_e+(\delta^t_x)_e$ for all $t=1,...,\tilde O(\sqrt{n})$ goes through.

The difficulty is that for any edge $e$, both $g_e$ and $\tau_e$ depend not just on $x_e$, but are a global property depending on the entire flow $x\in\R^E$. So we need to store additional information.

\ifdefined\arxiv
\subparagraph{Storing Weights $\tau\in\R^E$ (\Cref{sec:ipm:tau}):}
\else 
\subparagraph{Storing Weights $\tau\in\R^E$:}
\fi

For some $p = 1-1/\Theta(\log m)$, the (regularized) $\ell_p$-Lewis-weights from \cite{ls14} are defined as the scaling $\tau\in\R^E_{\ge0}$ that satisfy the recurrence relation
\begin{align}
    \tau = \sigma(\mT^{1/2-1/p} \mV^{-1/2} \mA) + \frac{n}{m} \notag
\end{align}
for $\mV_{e,e}=1/x_e^2-1/(u_e-x_e)^2$, $\mT_{e,e}=\tau_e$. \cite{bll+21} proved that the algorithm still works when the equality is only satisfied as some $(1\pm1/O(\log n))$-approximation.

Cohen and Peng \cite{CohenP15} proved that, if we start with some bad approximation $w$ of the Lewis-weights, then the following routine improves the approximation quality by a $1-p/2\approx 1/2$ factor
\begin{equation}
    w\gets(w^{2/p-1}(\sigma(\mW^{1/2-1/p}\mV^{-1/2}\mA)+n/m))^{p/2}. \label{eq:overview:taurecursion}
\end{equation}
Since $n/m\le \tau \le 1+n/m$, we can start with $w=\text{all-1-vector}$, and use some $\tilde O(1)$ iterations of \eqref{eq:overview:taurecursion} to get an accurate enough estimate of the Lewis-weight.
If we let $w^{(0)},w^{(1)},...,w^{(k)}=:\otau$ be the sequence of values computed this way, then the final $\otau$ will be a good approximation of the Lewis-weights.
We store these values in low space via the recurrence
\[
w^{(j)}_e = ((w^{(j-1)}_e)^{2/p-1}(\sigma((\mW^{(j-1)})^{1/2-1/p}\mV^{-1/2}\mA)_e+n/m))^{p/2}
\]
which allows us to query $w^{(j)}_e$ recursively for $j=1,...,k$, and thus store $\otau=:w^{(k)}\in\R^E$ in low space, if we can store $\sigma((\mW^{(j-1)})^{1/2-1/p}\mV^{-1/2}\mA) \in \R^E$ in low-space for $j=1,...,k=\tilde O(1)$.

This entire approach still yields an accurate estimate if we only approximately compute $\sigma(\cdot)$. This is done via the $\ell_2$-characterization of $\sigma(\cdot)$: for any matrix $\mM$ we have
\[
\sigma(\mM)_e := (\mM(\mM^\top \mM)^\dagger \mM^\top)_{e,e} = \|\mM(\mM^\top\mM)^{\dagger} \mM^\top\unit_e\|_2^2
\]
where $\unit_e$ is a standard unit-vector.
The $\ell_2$-norm can be $(1\pm\epsilon)$-approximated via an $O(\epsilon^{-2}\log n)\times|E|$ Johnson-Lindenstrauss sketching matrix $\mS$, i.e., we store for $j=1,2,...,k$
\[
\mP^{(j)}:=\mS^{(j)} (\mW^{(j-1)})^{1/2-1/p} \mV^{-1/2}\mA (\mA^\top (\mW^{(j-1)})^{1-2/p} \mV \mA)^{-1} \mA)^\dagger 
\]
where for $\epsilon = 1/\polylog(n)$ matrix $\mP^j$ is $\tilde O(1)\times n$ dimensional so takes only $\tilde O(n)$ space.
We can query for edge $e=(a,b)$ the value
\begin{align*}
\sigma((\mW^{(j-1)})^{1/2-1/p}\mV^{-1/2} \mA)_e 
\approx&~
\|\mP^{(j)} \mA^\top \mV^{-1/2} (\mW^{(j-1)})^{1/2-1/p}\unit_e\|_2 \\
=&~
\|\mP^{(j)} (\unit_b-\unit_a)\|_2^2 \cdot v_e^{-1} (w^{(j-1)}_e)^{1-2/p}.
\end{align*}
\begin{align}
w^{(j)}_e = ((w^{(j-1)}_e)^{2/p-1}(\|\mP^{(j-1)} (\unit_b-\unit_a)\|_2^2 \cdot v_e^{-1} (w^{(j-1)}_e)^{1-2/p}+n/m))^{p/2}. \label{eq:overview:tau:vertex}
\end{align}
Note that in addition to the $\tilde O(n)$ space for $\mP^{(1)},...,\mP^{(k)}$, here we need $v_e$ to compute $w^{(k)}_e$. We have access to that value since $v_e = 1/x_e^2-1/(u_e-x_e)^2$ and we already established that we can query $x_e$ when edge $e$ and its capacity $u_e$ are given to us.
While computation of $x_e$ requires $\tau_e$ and computing $\tau_e$ requires $x_e$, there is actually no circular dependency here.
That's because, if $x^t$,$\tau^t$ are the values computed during iteration number $t$ (i.e., $x^t = x^{t-1} + \delta_x^t$ where $\delta_x^t$ uses $\tau^t$),
then querying $x^t_e$ requires $\tau^t_e$, but querying $\tau^t_e$ needs $x^{t-1}_e$.
So after $t \le \tilde O(\sqrt{n})$ this loop terminates.
This does mean though, that we need to keep track/store all previous values of $\tau$.

In summary, we need $\tilde O(n)$ space per iteration to store $\tau$ while supporting queries to $\tau_e$, for a total of $\tilde O(n^{1.5})$ space over $\tilde O(\sqrt{n})$ iterations.

\ifdefined\arxiv
\subparagraph{Storing the flow $g$ (\Cref{sec:gradient}):}
\else 
\subparagraph{Storing the flow $g$:}
\fi
We here give only a very high-level description because details on how $g$ is defined require further details about the interior point method.
The exact choice of $g$ is crucial for the $\tilde O(\sqrt n)$ iteration count, and storing this flow in low space relies on recent developments in IPM-theory: for the classic Lee-Sidford IPM \cite{ls14}, it is unclear how to store $g$ in such low space, and our low space representation relies on the recently developed variants \cite{bll+21} that are robust against additional approximation errors.

In somewhat simplified terms, $g$ is defined w.r.t.~$x$, $\tau$ and reduced costs $s:=c-\mA \sum_t \delta_y^t$. (Note that for any edge $e=(w,z)$, we can also compute $s_e = c_e - \sum_t (\delta_y^t)_z - (\delta_y^t)_w$ since we already store the vertex potentials $\delta_y^t \in \R^V$ and cost $c_e$ is given at query time.)
Flow $g$ also has the property that if for two edges $e\neq e'$ we have $(\tau_e,x_e,s_e)=(\tau_{e'}, x_{e'},s_{e'})$, then $g_e=g_{e'}$ are the same.
Further, because of recent insights in robust interior point methods, we know the algorithm still works when $g$ is defined w.r.t.~approximations $\ox,\os,\otau$ that differ from the exact values by some $(1\pm\epsilon)$ approximation for $\epsilon=1/O(\log n)$.
Thus we can round each $\tau_e, x_e, s_e$ to multiples of $1+\epsilon$,
then $g\in\R^E$ contains only $\tilde O(1)$ distinct values.
We store these distinct values $\tilde g \in \R^{\tilde O(1)}$ in $\tilde O(1)$ space.

Then during a query, when receiving edge $e$, we (i) compute $x_e, s_e, \tau_e$, (ii) round them to the nearest multiple of $(1+\epsilon)$, (iii) access the corresponding $\tilde g_i$.

This is somewhat simplified, since the construction of $g$ from $x,s,\tau$ is nontrivial. However, those details require further discussion of how the IPM works, which are given in next subsection.
With these details, we can then also explain how to implement the construction of these representations in the semi-streaming model.
So far, we only discussed how to represent/store these values, but not how to compute them efficiently.
This is outlined in \Cref{sec:overview:streaming}, after giving details about the IPM in \Cref{sec:overview:reviewipm}.

\subparagraph{Remarks on 2-Party Communication}

The blackbox reduction from 2-party communication to streaming simulates the streaming algorithm. To simulate one pass over the input, party A passes over their own edge set, then sends their entire memory to party B who continues the current pass over their half of the edges.
This reduction would lead to only a trivial $\tilde O(n^2)$ communication bound as we have $\tilde O(\sqrt{n})$ iterations with $\tilde O(n^{1.5})$ space. 
However, observe that the $\tilde O(n^{1.5})$ space usage of our algorithm comes from storing each of $\delta_x^1,...,\delta_x^{\tilde O(\sqrt n)}$ in $\tilde O(n)$ space.
There is no need to send the entire memory to simulate the streaming algorithm in the 2-party model, because both parties already have the representation of $\delta_x^1,...,\delta_x^{t-1}$ when computing the representation of the new $\delta_x^t$ together when simulating the streaming algorithm. 
We will see in \Cref{sec:overview:streaming} that the streaming algorithm needs only $\tilde O(n)$ additional space and $\tilde O(1)$ passes to compute $\delta_x^t$. This implies $\tilde O(n)$ communication for each $\delta_x^t$, i.e., $\tilde O(n^{1.5})$ communication overall.

\subsection{Review of LS Central Path/Interior Point Method}
\label{sec:overview:reviewipm}

As context and motivation for our method, we start by discussing the method of \cite{ls14}. Readers familiar with the $\tilde O (\sqrt n )$ iteration interior point methods from \cite{ls14,bll+21} can skip to \Cref{sec:overview:streaming}.

For matrices $\mA \in \R^{m\times n}$, vectors $b\in\R^n$, $c\in\R^m$, and upper bounds $u\in\R^m$, this IPM solves linear programs of the form
\begin{equation}
\min_{\substack{x\in\R^m: \mA^\top x=b \\ 0\le x_e\le u_e \forall e\in[m]}} c^\top x. \label{eq:primal}
\end{equation}
The central path method of \cite{ls14}
repeatedly solves the following minimization problem. In each iteration, parameter $\mu\in\R_{>0}$ is decreased, and then the new minimizer $x_\mu$ is found by performing one Newton-step from the old minimizer $x_{\mu'}$.
\begin{equation}
    x_\mu:=\argmin _{\mA^\top x=b} f_\mu(x),\quad f_\mu(x) := c^\top x+\mu\sum_{i=1}^m \tau(x)_e\phi_e(x_e), \label{eq:central_path}
\end{equation}
where $\tau \in \R_{>0}^m$ and $\phi_e:\R_{>0}\to\R$ is defined as $\phi_e(z)=-\log(u_e-z)-\log(z)$. 
(For simplicity, we also write $\phi(x)\in\R^m$ for the vector with entries $\phi_e(x_e)$ for $e=1,..,m$.)
Note that as $x_e$ approaches 0 or capacity $u_e$, the $\phi_e(x_e)$ term goes to $\infty$, which keeps the minimizer $x_\mu$ within the feasible region. 
As $\mu \to 0$, the term $c^\top x$ starts to dominate; thus $x_\mu$ converges towards the optimal solution of the linear program for $\mu\to\infty$.
Here $\tau\in\R^m$ is a weight measuring the importance of each row of $\mA$.
\cite{ls14} chooses $\tau$ to be the regularized $l_p$ Lewis-weight for $p=1-\frac{1}{4\log(4m/n)}$. This is defined to be the solution to 
\begin{align}
    \tau=\sigma(\mT^{\frac{1}{2}-\frac{1}{p}}(\Phi''(x))^{-\frac{1}{2}}\mA) + \frac{n}{m}, \label{eq:overview:ipm:tau}
\end{align}
where $\Phi''(x)$ is a diagonal matrix with the 2nd derivatives $\phi''_e(x_e) = 1/x_e^2 - 1/(u_e-x_e)^2$ on the diagonal for $e=1,...,m$.

The curve $x_\mu$ for $\mu\in(\infty,0)$ is referred to as the central path.
It turns out that it is easy to find a point $x$ close to $x_\mu$ on this path for large $\mu$ 
\ifdefined\arxiv
(see Section~\ref{sec:initial_point:flow} for details). %
\else 
(see full version for details). %
\fi
Starting from this point, we iteratively take steps to decrease $\mu$ and move $x$ closer to the new minimizer $x_\mu$.

To do this, we must define a measure for how far $x$ is from $x_\mu$.
To satisfy $x=\argmin_{\mA^\top x =b} f_\mu(x)$ (\cref{eq:central_path}), optimality conditions tell us that we need $\mA^\top x=b$ and the gradient $\nabla f_\mu(x) = c + \mu \mT\phi'(x)$ must be orthogonal to the feasible hyperplane $\{x\mid \mA^\top x = b\}$; i.e., $c+\mA y+\mu \mT\phi'(x)=0$ for some $y\in\R^n$. 
We can rewrite this as $s=\mu\mT\phi'(x)=0$, where $s:=c+\mA y$. Now, to define what it means to follow the central path, we say that a triple $(x, s, \mu)$ is \textit{central} if
\begin{equation}
    \mA^\top x=b,\quad s=\mA y+c, \quad s+\mu\mT\phi'(x)=0. \label{eq:central}
\end{equation}
\noindent
As we cannot expect to be exactly centered, we measure centrality via the following potential function:
\begin{equation}
    \Psi\left(\frac{s+\mu\tau(x)\phi'(x)}{\mu\tau(x)\sqrt{\phi''(x)}}\right):=\sum_{e=1}^m \cosh\left(\lambda\left(\frac{s_e+\mu\tau(x)_e\phi'_e(x_e)}{\mu\tau(x)_e\sqrt{\phi''_e(x_e)}}\right)\right), \label{eq:centrality}
\end{equation}
where $\lambda = O(\log n)$. Note that the numerator is 0 (which is equivalent to minimizing \eqref{eq:centrality}) if and only if \eqref{eq:central} is satisfied.

Our goal is now to take steps where we alternate between decreasing $\mu$ (which increases \eqref{eq:centrality}) and updating $(x, y, s)$ such that \eqref{eq:centrality} decreases again. This way we approximately trace the curve $x_\mu$.

In the central path method by \cite{bll+21}, the improvements in each iterations are of the following form.
Given $(x^{t-1},s^{t-1},\mu^{t-1})$ at iteration number $t-1$, the next values for iteration $t$ are given by
\[
x^{t} \leftarrow x^{t-1} + \delta_x^{t},~~
s^{t} \leftarrow s^{t-1} + \delta_s^{t},~~
\mu^{t} \leftarrow \left(1-\frac{1}{\tilde O(\sqrt{n})}\right)\mu^{t-1}.
\]
To compute the movements $\delta_x^t,\delta_s^t$ we first calculate the weights $\tau^t\in\R^m$ of our current point $x^{t-1}$, and set the vector $g^t \in \R^m$ to be the gradient of \eqref{eq:centrality}, which intuitively moves $x$ as close as possible to $x_\mu$ on the central path. 
\cite{bll+21} have shown that instead of exact calculation of $\tau^t$, it suffices to compute weight $\otau^t$ that satisfy recurrence \eqref{eq:overview:ipm:tau} only with some $(1\pm 1/O(\log n))$-approximation.
Then, the movements are given by
\begin{align}
\mH^t \approx&_{1/O(\log m)} (\mA^\top(\oT^t)^{-1}\Phi''(x^{t-1})^{-1}\mA) \label{eq:steps}\\
\delta_x^{t} =&~ \Phi''(x^{t-1})^{-1/2}g^{t}-(\oT^{t})^{-1}\Phi''(x^{t-1})^{-1}\mA(\delta_y^{t}+\delta_c^{t})\notag\\
\delta_y^{t} =&~ (\mH^t)^{-1} \mA^\top \Phi''(x^{t-1})^{-1/2} g^{t} \notag\\
\delta_s^{t} =&~ \mu^{t} \mA \delta_y^{t} \notag\\
\delta_c^{t} =&~ (\mH^t)^{-1}(\mA^\top x^{t-1} - b) \notag\\
g^t = &~ \argmax_{w:~\|w\|_{\otau+\infty}} \langle w, -\nabla \Psi\left(\frac{s+\mu\otau(x)\phi'(x)}{\mu\otau(x)\sqrt{\phi''(x)}}\right) \rangle \notag
\end{align}
where $\mH^t$ is a spectral approximation, allowing us to solve the linear system approximately (i.e., we can use fast Laplacian system solvers).
Because of this approximation, we can no longer guarantee $\mA^\top x = b$, but inclusion of vector $\delta_c^t$ guarantees $\mA^\top x$ stays close to $b$.
Here $g^t$ is essentially the gradient of our potential function $\Psi$ (since we want to reduce that potential in each iteration), but we are taking a step that is bounded in a certain $(\tau+\infty)$-norm\footnote{$
    \|v\|_{\tau+\infty} := \|v\|_\infty + C \log(4m/n)\cdot\|v\|_\tau
    $
    where $\|v\|_\tau = (\sum_i \tau_i \cdot v_i^2)^{1/2}$ and $C$ is some large constant.} 
\ifdefined\arxiv
(\Cref{def:tauinftynorm}).
\else 
(see full version for details).
\fi
Calculation of this $g^t$ in low space will later be one of the main issues discussed in \Cref{sec:overview:streaming}.

In \cite{ls14} it was shown that $\tO(\sqrt{n})$ iterations of the central path method \eqref{eq:steps} suffice, i.e.~one must compute \eqref{eq:steps} for $t=1,2,...,\tO(\sqrt{n})$.
The naive implementation of \eqref{eq:steps} takes $\Omega(m)$ space which is $\Omega(n^2)$ in worst-case, since the vectors $x^t,s^t,g^t\in\R^m$ are $m$-dimensional. 
However, $O(n^2)$ space for min-cost flow is trivial, because then we could solve the problem in a single pass by storing the entire input graph. 
Therefore, our goal is to implement \eqref{eq:steps} using less than $O(n^2)$ space.

In the following, we outline how we can 
implement all calculations in \eqref{eq:steps} in the semi-streaming model.
In particular, we show that for any one $t$, \eqref{eq:steps} can be calculated in just $\tO(1)$ passes over the input.
Since we have $t=1,2,...,\tO(\sqrt{n})$, our algorithm takes $\tO(\sqrt{n})$ passes in total.

\subsection{Streaming Implementation}
\label{sec:overview:implement}
\label{sec:overview:streaming}

Our goal is to calculate \eqref{eq:steps} a total of $\tO(\sqrt{n})$ times in a semi-streaming setting with $\tilde O(1)$ passes over the input per iteration, and using only $\tO(n^{1.5})$ total space.

We already outlined in \Cref{sec:overview:graph} that we can store the values of the $t$-th iteration in total space $\tilde O(nt)$ by storing each iteration in $\tilde O(n)$ space.
The precise values being stored are:
\begin{itemize}
    \item $y^\ell, \delta_y^\ell,\delta_c^\ell \in \R^n$ for $\ell=1,...,t.$ This uses $O(nt)$ space.\\
    This implicitly represents $s^\ell\in\R^m$ because for any edge $e=(w,z)$ and its cost $c_e$, we can compute $s^\ell_e = c_e - (\mA y^\ell)_e = c_e - (y^\ell_{z} - y^\ell_w)$.
    \item We store $\tau^\ell$ for $\ell=1,...,t$ implicitly in $\tilde O(nt)$ space via a certain recursive JL-sketch.
    \item We store $g^t$ for $\ell=1,...,t$ implicitly in $\tilde O(t)$ space. (Details to be discussed further below).\\
    This implicitly represents $x^\ell\in\R^m$ because for any edge $e=(w,z)$ and its cost/capacity $c_e,u_e$, we can recursively compute 
    \[x^\ell_e = x^{\ell-1}_e + \phi''_e(x_e^{\ell-1})^{-1/2}g^\ell_e - (\otau_e^\ell)^{-1}\phi''_e(x^{\ell-1}_e)(\mA(\delta_y^\ell+\delta_c^\ell))_e.\]
\end{itemize}

The main remaining question is how to compute the above $\tilde O(n)$ representation for the next iteration $t+1$ in the streaming model with only $\tilde O(1)$ passes over the input.

\subparagraph{Lewis weights $\otau^t\in\R^m$}

In \Cref{sec:overview:graph} we established that in each iteration of the interior point method, the Lewis-weights are computed via $\tilde O(1)$ repetitions of
\begin{equation}
    w^{(j)}_e\gets((w^{(j-1)}_e)^{2/p-1}(\sigma((\mW^{(j-1)})^{1/2-1/p}\mM)_e+\frac{n}{m}))^{p/2}.
\end{equation}
where $\mM = \Phi''(x^t)\mA$. We start with $w^0=1$-vector, and assign $\otau^t = w^{(k)}$ for some $k=\tilde O(1)$. For each $t$, all $w^{(0)},...,w^{(k)}$ are stored, so that we can compute $w^{(j)}_e$ recursively from $w^{(j-1)}_e$ whenever an edge $e$ is received.

Here the main issue is calculating the $\tilde O(n)$ space representation of $\sigma((\mW^{(j-1)})^{1/2-1/p}\mM) \in \R^m$. 
We discussed in \Cref{sec:overview:graph} that this value can be represented via a Johnson-Lindenstrauss matrix $\mS^{(j)}$ and by computing
\begin{align}
\mS^{(j)} (\mW^{(j-1)})^{1/2-1/p} \mM (\mM^\top (\mW^{(j-1)})^{1-2/p} \mM)^{-1}. \label{eq:sketchfromleft}
\end{align}
Since an approximation suffices, we can replace $\mM^\top (\mW^{(j-1)})^{1-2/p} \mM$ by a spectral approximation, computed in one pass via \Cref{lem:spectralsparsifier} (for any given edge $e$, we have access to the corresponding row of $\mM = \Phi''(x^t)\mA$ because we have access to $x^t_e$.)\\
Once we have the spectral sparsifier, we can compute \eqref{eq:sketchfromleft} from left to right via one additional pass over the input.
This is explained in more detail in
\ifdefined\arxiv
\Cref{sec:ipm:tau}.
\else 
the full version.
\fi

\subparagraph{Constructing $y^t,\delta_y^t,\delta_c^t$}
These vectors are $n$-dimensional and can thus be computed/stored explicitly.
We have by \eqref{eq:steps} that
\begin{align}
\mH^t \approx&_{1/O(\log m)} (\mA^\top(\omT^t)^{-1}\Phi''(x^{t-1})^{-1}\mA) \\
\delta_y^{t} =&~ (\mH^t)^{-1} \mA^\top \Phi''(x^{t-1})^{-1/2} g^{t} \notag\\
\delta_c^{t} =&~ (\mH^t)^{-1}(\mA^\top x^{t-1} - b) \notag\\
y^t =&~ y^{t-1} + \delta_y^t \notag 
\end{align}
Note that when receiving edge $e$ on the stream, we can query $x^{t-1}_e$ (and thus $\Phi''(x^{t-1})_e$) and $\otau^t_e$ via previous paragraph.
Thus $\mH^t$ can be constructed in one pass over the input, by constructing spectral sparsifier via \Cref{lem:spectralsparsifier}. 
This takes $\tilde O(n)$ space.
We can then compute $\delta_y^t,\delta_c^t$ from right to left in one additional pass over the input.

\subparagraph{Constructing gradient $g^t\in\R^m$}

We are left with discussing how to construct vector $g^t$ from \eqref{eq:steps}, which can be seen as some gradient.
We uses a potential function of the form
$
\Phi(v^t) := \sum_{e \in [m]} \cosh(\lambda(v_e^t))
$
for some fixed $\lambda = O(\log m)$
and vector $v^t\in\R^m$ where
\[v^t_e:=\frac{s_e^t+\mu^t\otau_e^t\phi'_e(x_e^t)}{\mu^t\otau_e^t\sqrt{\phi''_e(x_e^t)}}\]
for each $e\in[m]$.
Here the gradient is
$
\nabla\Phi(v) := -\sum_{e \in [m]} \lambda\sinh(\lambda(v_e^t))
$. 
Observe that to obtain $(\nabla\Phi(v^t))_e$ for any $e\in[m]$, we only need to compute $v_e^t$.
Further, \cite{bll+21} proves that one does not need to compute the exact gradient $\nabla\Phi(v^t)$. Instead, it suffices to compute $\nabla\Phi(\tilde v^t)$ where $\tilde v^t_e$ is the appropriate vector in terms of $\tilde x^t,\tilde s^t,\tilde \tau^t$, which are each element-wise $(1\pm\epsilon)$ multiplicative approximations of $x^t,s^t,\tau(x^t)$ respectively for some $\epsilon = O(1/\log m)$.

To obtain $g^t$ from $\nabla \Phi(\tilde v^t)$, \cite{bll+21} defines %
\begin{align}
    g^t = \nabla\Psi(\tilde{v})^{\flat(\tilde{\tau})} := \argmax_{\|h\|_{\tilde\tau+\infty}=1}h^\top\nabla\sum_{e=1}^m\cosh(\lambda\tilde{v}_e). \label{eq:gradientoverview}
\end{align}
The maximization problem in \eqref{eq:gradientoverview} can be efficiently solved when rounding each entry of $\tilde v^t$ and $\tilde \tau^t$ to a multiple of some small $\delta=O(\epsilon)$. 
The additional approximation error from this rounding process can be charged to the approximation $\tilde \tau^t \approx \tau^t$, $\tilde v^t \approx v^t$.
Since the gradient is bounded, this discretization essentially reduces the optimization problem onto some smaller polylog-dimensional space, because we obtain at most some polylog distinct multiples of some $\epsilon=1/O(\log n)$.

Observe that within a single pass over the input, we can count how many entries of $\tilde v^t$ are rounded to the same value.
That is, within a single pass we compute two $\tilde O(1)$-dimensional vectors: $\tilde w^t$ containing the rounded values, and $\tilde k^t$ where $\tilde k^t_i$ counts how many entries have been rounded to $\tilde w^t_i$.
We can now solve \eqref{eq:gradientoverview} locally in $\tilde O(1)$ space, because (i) we know the output vector also has only $\tilde O(1)$ distinct values, (ii) we can compute the norm of the $m$-dimensional vector $h$ in \eqref{eq:gradientoverview} by using the counting vector (i.e., we know how often each entry will be duplicated).

\subparagraph{Summary, space efficient implementation.}

In each iteration of the central path method, we must compute \eqref{eq:steps}.
We have implicit access to $x^t,s^t$ which allows us to compute implicit access to $g^t$ and $\otau^t$.
We can then construct the spectral sparsifier $\mH^t$.
With these, we can then within a constant number of passes through the stream compute $\delta_y^{t+1},\delta_c^{t+1},y^{t+1}$.
We thus have all information that is required for the implicit representation of $x^{t+1},s^{t+1}$ so we can proceed with the next iteration.
After $\tO(\sqrt{n})$ iterations, our algorithm terminates. Since each iteration takes $\tO(1)$ passes, we take $\tO(\sqrt{n})$ passes in total.
The space requirement is $\tO(n^{1.5})$ since we store $y^t,\delta_c^t,\delta_y^t,\otau^t,g^t$ for $t=1,...,\tO(\sqrt{n})$.

\section{Low Space Implementation of IPM}
\label{sec:low_space}

In this section, we summarize the interior point method from \cite{bll+21} and show that it can be efficiently implemented in the streaming model. The following result states that, if we have access to an initial point via entry queries, then the interior point method can be implemented in $\tilde O(n^{1.5})$ space.

\begin{restatable}{theorem}{IPM}\label{thm:IPM}
    Given $\mu^\target \in \R_{>0}$ and $\epsilon$-centered (see \Cref{def:epscentered}) $(x^\init,s^\init,\mu^\init)$ in implicit representation (i.e., we can query any entry in time $T_0$), then in $\tilde O(T_0mn\log^2(\mu^\init/\mu^\target))$ time, $\tilde O (\sqrt{n} \log(\mu^\init/\mu^\target))$ passes, and $\tilde O (n^{1.5} \log(\mu^\init/\mu^\target))$ space, we construct $\epsilon$-centered $(x,s,\mu^\target)$ in implicit representation with query time $O(T_0+\sqrt{n} \log(\mu^\init/\mu^\target))$.
    The algorithm is randomized and the output is correct w.h.p.
\end{restatable}

Here, $\epsilon$-centered roughly means on the central path that the algorithm follows and is formalized in \Cref{def:epscentered}. %

\subsection{Overview of Analysis}
\label{sec:ipm:overview}
We will show how to implement the interior point method (\Cref{alg:short_step}) in low space in the streaming model.
The interior point method was previously outlined in \Cref{sec:overview:reviewipm}, in particular \eqref{eq:steps} stated the required calculations.
To implement this interior point method, we store several objects in memory.
We call the collection of these objects that are accessible during the $t$-th iteration of the interior point method the ``$t-1$-transcript of the \textsc{IPM}''.
\begin{definition}[$t-1$-transcript of the \textsc{IPM}]\label{def:objects}
We call the following collection of values the $t-1$-transcript of the \textsc{IPM}. Here $\cT(\cdot)$ is a function describing the time complexity of a query.
\begin{itemize}
    \item $\delta_y^1, \cdots, \delta_y^{t-1} \in \R^n$,  $\delta_c^1,\cdots,\delta_c^{t-1} \in \R^n$,  $y^1 , \cdots, y^{t-1} \in \R^n$ , $\mu^{1}, \cdots, \mu^{t-1} \in \R_{>0}$
    \item $g^1 , \cdots, g^{t-1} \in \R^m$ in implicit representation, i.e.~for any $e=(u,v)$, $1 < \ell < t$ we can query $g^1$ in time $T_0$ and $g^\ell_e$ in time ${\cal T}(t-1)$. %
    \item $x^1,\dots, x^{t-1}\in\R^m$ in implicit representation, i.e.~for any $e=(u,v)$, $1< l<t$, we can query $x^1$ in time $T_0$ and $x_e^l$ in time ${\cal T}(t-1)$.
    \item $s^1,\dots,s^{t-1}\in\R^m$ in implicit representation, i.e.~ for any $e=(u,v)$, $1\le l<t$, we can query $s_e^l$ in time $O(1)$.
    \item $\tau^1, \cdots,\tau^{t-1}\in\R^m$ in implicit representation, i.e.~for any $e=(u,v)$, $1 \le \ell < t$ we can query $\otau^\ell_e$ in time $T_{x}+\tilde O(1)$, where $T_x$ is the time needed to query $x^{t-1}_e$.\\
    In general, the time to query $\otau^\ell_e$ is $\cT(t-1)+\tilde O(1)$ time. However, for many of our use cases when we want to find $\tau_e^t$ we also want to calculate $x_e^{t-1}$ itself as well; in that case, we can temporarily save the value of $x^{t-1}_e$ and expedite query time for $\tau^t_e$ to just $\tilde O(1)$ time. 
\end{itemize}
\end{definition}

For the query complexity $\cT(\cdot)$ of the transcript (\Cref{def:objects}), we prove in \Cref{sec:ipm:next} that a transcript can be computed with the following query complexity:
\begin{restatable*}{proposition}{queryTime}\label{prop:querytime}
    ${\cal T}(t-1)=T_0+\tilde O(t)$, where $T_0$ is the time complexity of querying any $x^1_e$. %
\end{restatable*}

Note that the $t$-transcript is just a collection of most objects constructed within the first $t-1$ iterations of the interior point method (\textsc{ShortStep} in \Cref{alg:short_step}).
Hence the name \emph{transcript}, as it stores the values of previous iterations.
The main technical result is to show that we can perform iteration $t$ using the $t$-transcript. By storing the resulting $(t+1)$-transcript, this then inductively implies that we can perform the interior point method for any number of iterations.

The aim of the rest of this section is to show that we can indeed compute a $t$-transcript in low space under the streaming model.
The $t$-transcript requires us to compute the following values explicitly:
$(\delta_y^{t},\delta_c^{t}, y^{t}, \mu^{t})$, and we need to implicitly maintain $x^{t}$, $s^{t}$, $\tau^t$ and $g^t$.

In this analysis, many variables depend on other variables, either from this iteration or from the previous iteration of the transcript. To clarify the relationships as well as to confirm we have no circular dependencies, the following \Cref{fig:varmap} is a diagram of the order of variables calculated to arrive at the $t+1$-transcript from the $t$-transcript.

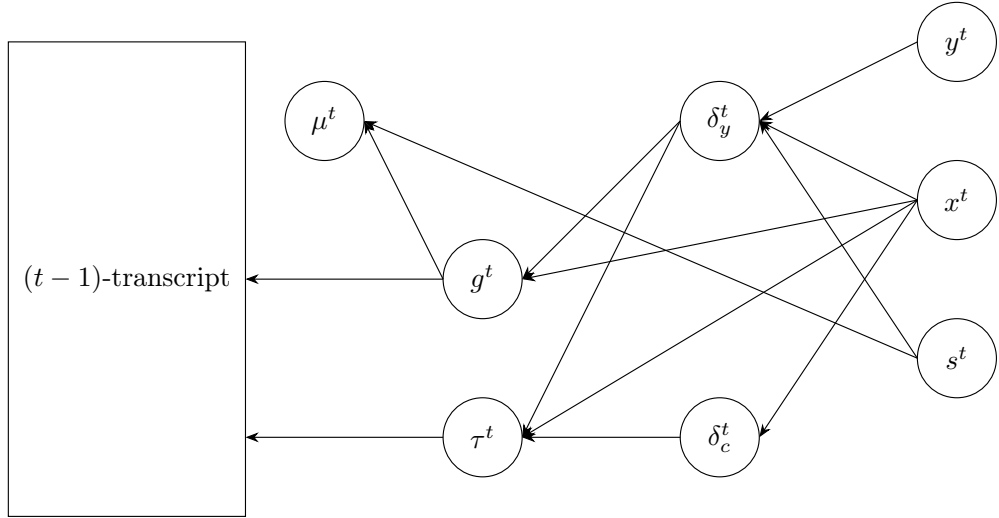
\begin{figure}[!ht]
\centering
\resizebox{0.8\textwidth}{!}{%
\begin{circuitikz}

\draw (1,0) rectangle (4,6)node[pos=0.5]{$(t-1)$-transcript}; 
\draw (7,1) circle (0.5cm) node { $\tau^t$} ;
\draw (7,3) circle (0.5cm) node { $g^t$} ;
\draw [->, >=Stealth] (6.5,1) -- (4,1);
\draw [->, >=Stealth] (6.5,3) -- (4,3);
\draw [->, >=Stealth] (6.5,3) -- (5.5,5);
\draw  (5,5) circle (0.5cm) node { $\mu^t$} ;
\draw  (10,5) circle (0.5cm) node { $\delta_y^t$} ;
\draw  (10,1) circle (0.5cm) node { $\delta_c^t$} ;
\draw  (13,6) circle (0.5cm) node { $y^t$} ;
\draw  (13,4) circle (0.5cm) node { $x^t$} ;
\draw  (13,2) circle (0.5cm) node { $s^t$} ;
\draw [->, >=Stealth] (9.5,5) -- (7.5,3);
\draw [->, >=Stealth] (9.5,5) -- (7.5,1);
\draw [->, >=Stealth] (9.5,1) -- (7.5,1);
\draw [->, >=Stealth] (12.5,6) -- (10.5,5);
\draw [->, >=Stealth] (12.5,4) -- (10.5,5);
\draw [->, >=Stealth] (12.5,4) -- (7.5,3);
\draw [->, >=Stealth] (12.5,4) -- (7.5,1);
\draw [->, >=Stealth] (12.5,4) -- (10.5,1);
\draw [->, >=Stealth] (12.5,2) -- (5.5,5);
\draw [->, >=Stealth] (12.5,2) -- (10.5,5);
\end{circuitikz}
}%

\caption{Variable dependencies. E.g., querying $x^t_e$ requires the value of $\tau^t_e$. \label{fig:varmap}}
\end{figure}

\subsection{Recap of IPM}

\begin{algorithm}[!ht]\caption{Interior Point Method, Algorithm 1 in \cite{bll+21}
}\label{alg:short_step}
\begin{algorithmic}[1]
\Procedure {IPM}{$\mA$, $x^{\init}$, $s^{\init}$, $\tau^{\init}$, $l$, $u$, $\mu$, $\mu^{\final}$}
	\State Fix $\tau(x) := w(\phi''(x)^{-\frac{1}{2}})$, as defined in \Cref{def:lewisweights}.
 	\State Define $\alpha=\frac{1}{4\log(4m/n)}$, $\epsilon = \frac{\alpha}{C}$, $\lambda=\frac{C\log(Cm/\epsilon^2)}{\epsilon}$, $\gamma=\frac{\epsilon}{C\lambda}$, $r=\frac{\epsilon\gamma(1-p)}{C\sqrt{n}} = \frac{\epsilon^2\gamma}{\sqrt{n}}$.
	\State $t\gets 1$, $x^t\gets x^{\init}$, $s^t\gets s^{\init}$, $\tau^t\gets \tau^\init$, $\mu^t\gets \mu$
	\While{$\mu^t>\mu^{\final}$}
		\State $\mu^{t+1}\gets(1-r)\mu^t$
		\State $(x^{t+1}, s^{t+1}, \tau^{t+1})\gets \Call{ShortStep}{x^t, s^t, \tau^t, \mu^t, \mu^{t+1}}$
		\State $t\gets t+1$
	\EndWhile
	\State \Return $(x^t, s^t)$
\EndProcedure
\Procedure {ShortStep}{$x$, $s$, $\tau$, $\mu$, $\mu^{\text{(new)}}$}
    \State Find $\otau$ such that $\otau\approx_\epsilon \tau(x)$
	\State Let $v=\frac{s+\mu\tau(x)\phi'(x)}{\mu\tau(x)\sqrt{\phi''(x)}}$ and let $\|\ov-v\|_{\infty}\le\gamma/20$.
	\State Let $g=-\gamma\nabla\Psi(\ov)^{\flat(\otau)}$ where \[\nabla\Psi(\ov)^{\flat(\otau)}:=\argmax_{\|h\|_{\tau+\infty}=1}h^\top\nabla\Psi(\ov)\] where $\Psi(\ov)=\sum_{i=1}^m\cosh(\lambda\ov)$
	\State Let $\mH^t\approx_{\gamma}=\mA^\top\omT^{-1}\Phi''(x)^{-1}\mA$
	\State Let $\delta_y^{t}:=\mH^{-1}\mA^\top\Phi''(x)^{-\frac{1}{2}} g$, $\delta_c^{t}:=\mH^{-1}(\mA^\top x-b)$
	\State $\delta_x^{t}\gets \Phi''(x)^{-\frac{1}{2}}g-\omT^{-1}\Phi''(x)^{-1}\mA(\delta_y^{t+1}+\delta_c^{t+1})$
	\State $\delta_s^{t}\gets \mu\mA\delta_y^{t+1}$
	\State $x^{t}\gets x+\delta_x^{t}$ and $s^{t}\gets s+\delta_s^{t}$.
	\State \Return $(x^{t}, s^{t}, \tau^{t})$.
\EndProcedure
\end{algorithmic}
\end{algorithm}

We start by reviewing the important definitions and lemmas that \cite{bll+21} uses in their interior point method.

The barrier function \cite{bll+21} uses is given by $\phi(x)=-\log(x-l)-\log(u-x)$. Note $\phi'(x)=-\frac{1}{x-l}+\frac{1}{u-x}$ and $\phi''(x)=\frac{1}{(x-l)^2}-\frac{1}{(u-x)^2}$.

During the IPM, we maintain tuples $(x, s, \mu)$. Given a current point $(x, s, \mu)$, we define a \textit{weight function} $\tau:\R^m\to \R^m_{>0}$. We use regularized Lewis weights.

\begin{definition}[Lewis weights]\label{def:lewisweights}
    For $p=1-\frac{1}{4\log(4m/n)}$ define the Lewis weights for $x$ to be the solution to
    \[
        \tau(x)=\sigma(\mT^{\frac{1}{2}-\frac{1}{p}}\Phi''(x)\mA)+\frac{n}{m} \text{ where } \mW := \diag(\tau).
    \]
\end{definition}

We now formally define the notion of following the central path. We say our current iterate are $\epsilon$-centered if the following condition holds. Throughout the algorithm, the points will stay $\epsilon$-centered. In particular, the initial point must be $\epsilon$-centered.

\begin{definition}[$\epsilon$-centered point]\label{def:epscentered}
    We say that $(x,s,\mu)\in\R^m\times\R^m\times\R^m_{>0}$ is $\epsilon$-centered for $\epsilon\in(0,1/80]$ if the following properties hold, where $C_{\mathrm{norm}}=C(1/p)$ for a constant $C\ge100$.
    \begin{enumerate}
        \item (Approximate centrality) $\left\|\frac{s+\mu\tau(x)\phi'(x)}{\mu\tau(x)\sqrt{\phi''(x)}}\right\|_\infty\le\epsilon$.
        \item (Dual Feasibility) There exists a vector $z\in\R^n$ with $\mA z+s=c$.
        \item (Approximate Feasibility) $\|\mA^\top x-b\|_{\mA^\top(\mT(x)\Phi''(x))^{-1}}\le\epsilon\gamma/C_{\mathrm{norm}}$.
    \end{enumerate}
\end{definition}

The following potential is used to guarantee $\epsilon$-centrality. The interior point method essentially performs a gradient descent w.r.t.~this potential.

\begin{definition}[Centrality Potential]
    We track the following centrality potential.
    \[
        \Psi(x,s,\mu):=\sum_{i=1}^m\psi\left(\frac{s_e+\mu\tau(x)_e\phi'_e(x_e)}{\mu\tau(x)_e\sqrt{\phi''_e(x_e)}}\right)
    \]
    for $\psi(x):=\cosh(\lambda x)$, where $\lambda=\Theta(\log(m)/\epsilon)$ (see \Cref{alg:short_step} for exact definition of $\lambda$).
\end{definition}

The interior point method sometimes measures the length of vectors in a mixed norm, defines as
\begin{definition}[{\cite[Definition 4.9]{bll+21}}]\label{def:tauinftynorm}
    For a vector $v\in\R^m,\tau\in\R^m_{\ge0}$, we define
    $$
    \|v\|_{\tau+\infty} := \|v\|_\infty + C \log(4m/n)\cdot\|v\|_\tau
    $$
    where $\|v\|_\tau = \sqrt{\sum_i \tau_i \cdot v_i^2}$ and $C$ is some large constant.
\end{definition}

We can now state the precise statement of the IPM of \cite{bll+21}, given by \Cref{alg:short_step} and \Cref{lem:ipm}. 

\begin{lemma}[{\cite[Lemma 4.12]{bll+21}}]\label{lem:ipm}
Consider a call to \textsc{IPM} (\Cref{alg:short_step}).
Let $x^\init$, $s^\init$, $\mu$ be such that $\mA z+s^\init=c$ for some $z$, let $\tau^\init$ be such that $\|T(x^\init)^{-1}(\tau^\init=\tau(x^\init))\|_\infty\le\epsilon$, and let $\epsilon,\gamma$ be as defined in \Cref{alg:short_step}, and $\left\|\frac{s+\mu\tau(x^\init)\phi'(x^\init)}{\mu\tau(x^\init)\sqrt{\phi''(x\init)}}\right\|_{\infty}\le \epsilon$ and $\|\mA^\top x^\init -b\|_{(\mA^\top(\mT(x^\init)\Phi''(x^\init))^{-1}\mA)^{-1}}\le \epsilon\gamma(1-p)/C$.
Then \Cref{alg:short_step} returns
$x,s \in \R^m$ such that $A y + s = c$ for some $y\in\R^n$ and $c^\top x-\min_{\substack{\mA^\top x=b\\ l_i\le x_i\le u_i \forall i \in [m] }}c^\top x\lesssim n\mu$.
The algorithm takes $\tO(\sqrt{n}\log (\mu/\mu^\final))$ iterations.
Throughout all iterations, $\left\|\frac{s+\mu\tau(x)\phi'(x)}{\mu\tau(x)\sqrt{\phi''(x)}}\right\|_{\infty}\le\epsilon$.
\end{lemma}

\subsection{Implicit Representation of \texorpdfstring{$\tau$}{}}\
\label{sec:ipm:tau}
Recall that the first step of each iteration is to find the Lewis weights $\tau(x)$ of our new point $x$ where $\tau(x)$ is defined as the solution to
\[
    \tau(x)=\sigma(\mT^{1/2-1/p}\Phi''(x)\mA)+\frac{n}{m}
\]
where $\mT$ is the diagonal matrix with $\tau$ on the diagonal.

Specifically, in the step of the algorithm, we want to maintain an approximate $\otau\approx_\epsilon \tau(x)$. Here, we in fact maintain an $\otau \approx_{\epsilon_\tau}\tau(x)$ for $\epsilon_\tau:=\gamma/60<\epsilon$ where $\gamma$ and $\tau$ are globally defined above in the interior point method (\Cref{alg:short_step}). 
We here show how to compute such $\otau$. Note that for our choice of accuracy $\epsilon_\tau$, it is logarithmic in $n$, and thus can be subsumed into the $\tilde O$s.

\begin{theorem}\label{thm:compute_tau}
    There is a randomized algorithm that receives as input: $p\in [\frac{1}{2},2)$, a directed weighted graph $G=(V,E)$ given as stream, $x\in\R^E_{>0}$ with query time $T_x$, and error parameter $\epsilon_\tau$.

    Then, within $\tilde O(mT_x+\epsilon_\tau^{-2})$ time, $\tilde O(1)$ passes through the stream, and with $\tilde O(n\epsilon_\tau^{-2})$ space, the algorithm constructs w.h.p.~an implicit representation of $\otau$ with $\otau\approx_{\epsilon_\tau} \tau(x)$
    . Each edge query can be done in $T_x+\tilde O(1)$ time. 
\end{theorem}

To prove \Cref{thm:compute_tau}, we first use the following lemmas.

\begin{lemma}[{\cite[Lemma 5.3]{bll+21}}]\label{lem:lemma_5.3_bll+21}
    Let $w$, $z$, $\osigma\in\R^m_{>0}$, $\epsilon,\epsilon_\sigma>0$ and $p\in(0,2)$ with $w\approx_{\epsilon}\sigma(\mW^{1/2-1/p}\mA)+z$ and $\osigma\approx_{\epsilon_\sigma}\sigma(\mW^{1/2-1/p}\mA)+z$. Define $w'=(w^{2/p-1}\osigma)^{p/2}$, then $w'\approx_{(\epsilon+\epsilon_\sigma)|1-p/2|+\epsilon_\sigma}\sigma(\mW'^{1/2-1/p}\mA)$.%
\end{lemma}

\begin{lemma}[{\cite[Lemma 5.5]{bll+21}}]\label{lem:lemma_5.5_bll+21}
    For any $w,z\in\R_{>0}^m$, $\epsilon>0$, and $p\in(0,2]$ with $w\approx_\epsilon\sigma(\mW^{\frac{1}{2}-\frac{1}{p}}\mA)+z$ we have $w\approx_{\epsilon}\tau(\mA)$.
\end{lemma}

\begin{proof}[Proof of \Cref{thm:compute_tau}]%
    The algorithm starts by picking some initial candidate for $\otau$, and then iteratively refines it. Let us start with the refinement procedure.
    \paragraph{Refinement}
    Suppose we already have some candidate $\otau\approx\tau(x)$ but with an approximation factor that is worse than $\epsilon_\tau$. The general approach is to iteratively improve the approximation factor. If for $\epsilon_\sigma\le\epsilon_\tau/10$ we had a $\osigma\approx_{\epsilon_\sigma}\sigma(\omT^{1/2-1/p}\Phi''(x)\mA)$, then we can then perform the following update,
    \begin{equation}\label{eq:tau_update}
        \otau\gets(\otau^{2/p-1}(\osigma_e+\frac{n}{m}))^{p/2},
    \end{equation}
    which by \Cref{lem:lemma_5.3_bll+21} would decrease the approximation factor by about a $1-2/p$ factor.
    
    It remains to find such an $\osigma$. To do so, we take, for each edge $e$, 
    \[
        \osigma_e=\|\mR \omT^{1/2-1/p}\Phi''(x)^{-1/2} \mA \mM^{-1}\mA^\top\Phi''(x)^{-1/2}\omT^{1/2-1/p}\chi_e\|_2^2, 
    \]
    where $\mR$ is a random sparse Johnson-Lindenstrauss matrix giving an $\epsilon_\sigma/2$ approximation, and $\mM$ is a $(1\pm \epsilon_\sigma/2)$-spectral sparsifier for $\mA^\top\Phi''(x)^{-1/2}\omT^{1-2/p} (x)^{-1/2} \mA$, constructed by \Cref{lem:spectralsparsifier}. In each iteration, we store $\mP:=\mR \omT^{1/2-1/p}\Phi''(x)^{-1/2} \mA \mM^{-1}$, which is a $O((\epsilon_{\sigma}/2)^{-2} \log n)\times n$ matrix. 

    For convenient of writing proofs, we define a local variable $\mB$ (a weighted version of $\mA$, only used in this proof) and the respective $\mQ := \mB^\top \mB$,
    \begin{align*}
        \mB := & ~ \omT^{1/2-1/p}\Phi''(x)^{-1/2}\mA, \\
        \mQ := & ~  \mA^\top\Phi''(x)^{-1/2}\omT^{1-2/p}\Phi''(x)^{-1/2}\mA .
    \end{align*}

    To see why matrix $\mP$ allows us to represent the leverage scores $\osigma$ in small space, recall we wanted $\osigma\approx_{\epsilon_\tau/10}\sigma(\omT^{1/2-1/p}\Phi''(x)\mA)$. Indeed,
    \begin{align*}
        & ~ ~~ ~ \sigma(\omT^{1/2-1/p}\Phi''(x)^{-1/2}\mA)_{e,e}\\
        & ~ = \sigma( \mB )_{e,e} \\
        & ~ = \chi_e^\top \mB (\mB^\top \mB)^{-1} \mB^\top \chi_e \\ %
        & ~ = \chi_e^\top \mB \mQ^{-1} \mB^\top \chi_e \\
        & ~ = \chi_e^\top \mB \mQ^{-1} \mQ \mQ^{-1} \mB^\top \chi_e \\
        & ~ = \chi_e^\top \mB \mQ^{-1} \mB^\top \mB \mQ^{-1} \mB^\top \chi_e \\
        & ~ = \| \mB \mQ^{-1} \mB^\top \chi_e \|_2^2 \\
        & ~ \approx_{\epsilon_\sigma/2} \| \mB \mM^{-1} \mB^\top \chi_e \|_2^2 \\ %
        & ~ \approx_{\epsilon_\sigma/2}  \|\mR \omT^{1/2-1/p}\Phi''(x)^{-1/2} \mA\mM^{-1} \mB^\top \chi_e\|_2^2.
    \end{align*}
    where the first step follows from definition of $\mB$, the second step follows from definition of leverage score, the third step follows from $\mQ = \mB^\top \mB$, the fourth step follows from $\mQ \mQ^{-1} = I $,  the fifth step follows from definition of the $\mQ = \mB^\top \mB$ norm, the seventh step follows from definition of $\ell_2$ norm, the eighth follows from $\mM$ being a $(1\pm \epsilon_\sigma/2)$-spectral sparsifier of $\mQ$, the last step follows from JL property.

    \paragraph{Initial $\otau$ for refinement}
    The above procedure allows us to improve some $\otau$, but we need a starting candidate. Since $1+n/m\ge\tau(x)\ge\frac{n}{m}$, we can simply take $\otau=1$, which is a $(m/n)$-approximation. Then by \Cref{lem:lemma_5.3_bll+21}, to get a $\otau\approx_{\epsilon_\tau}\sigma(\omT^{1/2-1/p}\mA)+\frac{m}{n}$, we simply need to calculate the update
    \eqref{eq:tau_update} a total of $\log_{|1-p/2|}((m/n)(\epsilon_\tau/2)^{-1}) = \tilde O(1)$ times (where our approximation of the leverage scores $\epsilon_\sigma$ is negligibly small). Finally, by \Cref{lem:lemma_5.5_bll+21}, we have $\otau\approx_{\epsilon_\tau}\tau(x)$.

    \paragraph{Query Complexity}
    We now analyze the query complexity per edge. Throughout this process we store all the $\mP$s that we have used to get from 1 to our final $\otau$. For every edge $e$, we only need to calculate $\|\mP\mA\Phi''(x)^{-1/2}_{ee}\omT'^{1/2-1/p}_{ee}\chi_e\|_2^2$ where $\otau'$ is the previous $\otau$ calculated. We can temporarily save $\Phi''(x)^{-1/2}_{ee}$ when calculating this edge, so to calculate this we only need to track back to every $\otau_e$ we have seen, of which there are $\tilde O(1)$ of. This takes $T_x+\tilde O(1)$ time.

    \paragraph{Construction Complexity}
    We now analyze resource usage. We take $\tilde O(1)$ steps. In each step we use 1 pass through the stream, and query the previous $\otau$, which as above takes $\tilde O(T_x)$ time per edge and $\tilde O(mT_x)$ time total. We solve for $\mM$ using \Cref{lem:spectralsparsifier}, where querying the weight on each edge takes time $T_x+\tilde O(1)$, for a total of $\tilde O(mT_x)$. Thus in total we use $\tilde O(1)$ passes through the stream and $\tilde O(mT_x+\epsilon_\tau^{-2})$ runtime. For storage, we store $\tilde O(1)$ variations of our matrix $\mP$, which takes $\tilde O(n\epsilon_\tau^{-2})$ space, completing the proof. 
\end{proof}

\subsection{Implicit Representation of \texorpdfstring{$g$}{}}
\label{sec:ipm:gradient}
Our next step is to find the direction we should move our new point $x$ to, which is given by a gradient in the direction that maximizes centrality. Recall that centrality is defined as
\[
    v=\frac{s+\mu\tau(x)\phi'(x)}{\mu\tau(x)\sqrt{\phi''(x)}}.
\]

The goal of this subsection is to show the following.

\begin{theorem}\label{thm:compute_g}
    Given a $t-1$-transcript, within one pass over the input graph, we compute and store an exact $g=-\gamma\nabla\Psi(\ov)^{\flat(\otau)}$ where \[\nabla\Psi(\ov)^{\flat(\otau)}:=\argmax_{\|h\|_{\otau+\infty}=1}h^\top\nabla\Psi(\ov)\] 
    \[\|\ov-v^t\|_\infty\le\gamma/20\]
    \[\otau \approx_\epsilon \tau^t\]
    where $\Psi(\ov)=\sum_{e=1}^m\cosh(\lambda\ov_e)$ in implicit representation. This takes $\tilde O(\log m)$ space, $\tilde O(m \cT(t-1))$ time.
    
    Given edge $e$, querying any entry $g_e$ takes ${\cal T}(t-1)+\tilde O(1)$ time.
\end{theorem}

Note that although $g$ is exact for $\ov$ and $\otau$, $\ov$ and $\otau$ are not. We only require $\|\ov-v\|_\infty\le\gamma/20$ and $\otau\approx_\epsilon \tau$. Thus to store $g$ efficiently we define the following approximations. Specifically, the goal of these is to project $g$ onto a smaller polylog dimensional space.

\begin{definition}
We use $\round_{1+\epsilon'}(z)$ to denote the closest $(1+\epsilon')^i$ for $z$, we use $\round_{\epsilon'}(z)$ to denote the closest $j \cdot \epsilon'$ for $z$. 
\end{definition}

Our vector $g^t$ will be computed using the following \Cref{def:wt_g}.

\begin{definition}\label{def:wt_g}
Given $\epsilon',\mu > 0$, $x^{t-1},s^{t-1},\otau^t\in\R^m_{>0}$
we define $\tilde g^t \in \R^m$ as:
\begin{align*}
    \tilde g^t := (\nabla \Psi(\round_{\epsilon'}(\frac{s^{t-1}+\mu \otau^t\phi'(x^{t-1})}{\mu \otau^t\sqrt{\phi''(x^{t-1})}})) )^{\flat(\round_{1+\epsilon'}(\otau^t))}.
\end{align*}
\end{definition}

To calculate this, we use the following theorem.

\begin{restatable*}{theorem}{thmGradient}\label{thm:g_construction}
    Given an $\epsilon > 0$, $f:\R\to\R$, a directed weighted graph $G=(V,E)$ given as stream, and implicit vectors $v,\tau \in\R^E$
    satisfying $n/m \le \tau \le 1$ and $-1\le v\le1$
    such that for any given $e\in E$ we can query $v_e$ and $\tau_e$
    in $T$ time.
    
    Then we can construct an implicit representation of target vector $g=(f(\round_{\epsilon'}(v))^{\flat(\round_{1+\epsilon'}\tau)}$.
    Our construction
    takes $O(\epsilon'^{-2} \log m)$ space and $O(mT)$ time.
    
    Given edge $e$, querying any entry $g_e$ takes $T+O(1)$ time.%
\end{restatable*}

With these, we can now prove our main result for this section.

\begin{proof}[Proof of \Cref{thm:compute_g}]
    We want to compute a $g = (\nabla\Psi(\ov))^{\flat(\otau)}$ where $\|\ov-v^t\|_\infty \le \gamma/20$ and $\otau\approx_\epsilon \tau^t$.
    We first claim that $\tilde g^t$ as defined in \Cref{def:wt_g} is a suitable choice for this $g$, when using $\epsilon'=\gamma/60$. 
    For this choice of $g$, we have $\ov=\round_{\epsilon'}(\frac{s^{t-1}+\mu \otau^t\phi'(x^{t-1})}{\mu \otau^t\sqrt{\phi''(x^{t-1})}})$. This satisfies $\|\ov-v^t\|_\infty \le \gamma/20$, because $x^{t-1}$ and $s^{t-1}$ are exact and $\otau^t\approx_{\gamma/60}\tau(x^{t-1})$ by \Cref{thm:compute_tau}.
    
    We want to construct this $g$ via \Cref{thm:g_construction} using $f=\nabla\Psi$, which requires 
    $$-1\le \underbrace{\frac{s^{t-1}+\mu \otau^t\phi'(x^{t-1})}{\mu \otau^t\sqrt{\phi''(x^{t-1})}}}_{=:v} \le1.$$ %
    This is satisfied by centrality guarantee of the IPM (Lemma~\ref{lem:ipm}).

    We now analyze the runtime of \Cref{thm:g_construction}. To query $v_e$ requires $x^{t-1}_e$, $s^{t-1}_e$, $\otau^{t}_e$, and $O(1)$ calculation time. The first can be queried in time ${\cal T}(t-1) $, while the second is $\tilde O(1)$. We can temporarily save the value of $x^{t-1}_e$, so that when we use \Cref{thm:compute_tau}, $\tau_e$ requires time $\tilde O(1)$ as our input parameter $T_x$ would be constant time. 
    Thus parameter $T$ in \Cref{thm:g_construction} is $O(\cT(t-1))$, which leads to a total runtime of $ O(m {\cal T}( t-1) )$. 
    Querying $g_e$ for an edge $e$ takes time ${\cal T}(t-1)+ O(1)$ by \Cref{thm:g_construction}. 
\end{proof}

\subsection{Implicit Representation of \texorpdfstring{$x$}{}, \texorpdfstring{$s$}{}}
\label{sec:ipm:primal}
The main result of this subsection is the following lemma, which states that we have access to $x^t_e$ and $s^t_e$ for any given edge $e$. Note that we do not have an update time for these, as we do not store anything that needs to be updated; these are computed ``on the fly.''
\begin{lemma}\label{cor:access_s}
    Given the $t-1$-transcript of the \textsc{IPM}, and an edge $e=(u,v)$ and its cost $c_e$, we can query $s^t_e$ explicitly in $O(1)$ time. 
\end{lemma}
\begin{proof}
    By \Cref{def:objects} of the $t$-transcript, we have $y^1,...y^t$ explicitly in memory.
    The slack is defined as $A y^t + s^t = c$, so for edge $e=(u,v)$ we have $s^t_e = c_e - (y^t_u - y^t_v)$ for all $t=1,...,t$.
\end{proof}
\begin{lemma}\label{cor:access_x}
    Given the $t-1$-transcript of the \textsc{IPM}, and an edge $e=(u,v)$ and its cost $c_e$, we can query $x^t_e$ explicitly in ${\cal T}(t-1)+\tilde O(1)$ time. 
\end{lemma}
\begin{proof}

    To find $x^t$, recall that we store
$\delta_c^1, \delta_c^2, \cdots, \delta_c^{t}$
and
$\delta_y^1, \delta_y^2, \cdots, \delta_y^{t}$
Note that
\begin{align*}
(\delta_x^{t})_e = (\Phi''(x^{t-1})^{-\frac{1}{2}}g^{t}_e-(\omT^{t})^{-1}\Phi''(x^{t-1})^{-1}\mA(\delta_y^{t}+\delta_c^{t}) )_e
\end{align*}

\noindent 
Using the object $(\delta_x^t)_e$ and given $x_e^{t-1}$, we can compute
$
(x^t)_e = x^{t-1}_e+  (\delta_x^t)_e
$. 
Thus, we complete the proof.

The above equation requires querying objects $x^{t-1}_e$, $g^{t}_e$, $\tau^{t}$. Querying $x^{t-1}$ takes ${\cal T}(t-1)$ time, and by temporarily saving this value we can query both $g^t_e$ and $\tau^t$ in $\tilde O(1)$ time. Thus the total runtime for every edge query is ${\cal T}( t-1) +\tilde O(1)$.
\end{proof}

\subsection{Constructing matrix \texorpdfstring{$\mH$}{}, dual \texorpdfstring{$\delta_y^{t}$}{} and \texorpdfstring{$y^t$}{}, and artificial variable \texorpdfstring{$\delta_c^t$}{}}
\label{sec:ipm:dual}
We now describe how we can construct the $(t+1)$-transcript when given the $t$-transcript.
The following lemma describes all the remaining necessary values.

\begin{lemma}\label{lem:induction_step_1}
Given the $t$-transcript of the \textsc{IPM} (\Cref{def:objects}),
then we can compute the following objects w.h.p.~in $O(1)$ passes (with polynomial time).
 Part 1. $\mH^{t}$, Part 2. $\delta_y^{t}$, Part 3. $\delta_c^{t}$, Part 4. $y^{t}$, Part 5. $\mu^{t}$. 
These additional values can be stored in $\tO(n)$ additional space.
\end{lemma}
\begin{proof}

{\bf Proof of Part 1.} 
To construct $\mH^t$. %
We use \Cref{lem:spectralsparsifier}, which tells us we need one pass over the input. For each edge, we need to find the edge weight $\phi_e(x^{t-1}_e)^{-1}(\tau^t_e)^{-1}$. 
These are given via \Cref{cor:access_x} and \Cref{thm:compute_tau}. The former requires time $\mathcal{T}(t-1)$, and the latter $\tilde O(1)$ if we temporarily save our value of $x^{t-1}_e$. Thus we can compute this in time $O(m {\cal T}(t-1) )$.

{\bf Proof of Part 2.} 
$
\delta_y^{t} = (\mH)^{-1} \mA^\top (\Phi''(x^{t-1})^{-\frac{1}{2}}) g^t
$. 
For this, first compute $\mA^\top (\Phi''(x)^{-\frac{1}{2}}) g$ in one pass and time $O(mT^{t-1})$. 
(This again requires us to compute $x_e$ and $g_e$ when we receive edge $e$ via \Cref{cor:access_x} and \Cref{thm:compute_g}.)
We then solve the Laplacian system on spectral sparsifier $\mH^t$ in time $O(m)$.

{\bf Proof of Part 3.}
We have $\delta_c = (\mH)^{-1}(\mA^\top x - b)$.
For this, we use one pass to compute $\mA^\top x -b$, then solve the Laplacian system on the spectral sparsifier, again using total time $O(m {\cal T}(t-1) )$.

{\bf Proof of Part 4.}
Compute and store %
$
y^{t} = y^{t-1} + \delta_y^{t}.
$

{\bf Proof of Part 5.} 
By definition from \Cref{alg:short_step}, $\mu^{t}$ is just $\mu^{t}=(1-r)\mu^{t-1}$
Therefore, we complete the proof.
\end{proof}

\subsection{Constructing the Next Transcript}
\label{sec:ipm:next}

We first show the following
\queryTime
\begin{proof}
    We use induction over the number of iterations. Given the first iteration takes time $T_0$, and \Cref{thm:compute_g,cor:access_x} tell us ${\cal T} (t) ={\cal T}( t-1)+\tilde O(1)$, the conclusion follows.%
\end{proof}

Then, our algorithm and its correctness follows direction from the above theorems:

\begin{algorithm}[ht]\caption{Our Streaming Algorithm 
}\label{alg:short_step_compress}
\begin{algorithmic}[1]
\Procedure {ShortStep}{$t$-transcript $\mathrm{Trans}^t$}
    \State Construct implicit representation of $\otau\approx_{\epsilon} \tau(x^t)$ via \Cref{thm:compute_tau}.
    \State // This takes $\tilde O(1)$ pass, $\tilde O(n)$ space, $\tilde O(m(t+T_0))$ time.
    \State Construct implicit representation of $g=-\gamma\nabla\Psi(\ov)^{\flat(\otau)}$ via \Cref{thm:compute_g}.
    \State // This takes 1 pass, $\tilde O (1)$ space, $\tilde O(m(t+T_0))$ time.
	\State Construct spectral sparsifier $\mH \approx_\gamma \mA^\top\omT^{-1}\Phi''(x)^{-1}\mA$ using \Cref{lem:spectralsparsifier}.
    \State // This takes 1 pass, $\tilde O(n)$ space, $\tilde O(m(t+T_0))$ time.
    \State Compute $\mA^\top\Phi''(x)^{-\frac{1}{2}} g$ (using implicit representation of $g$, \Cref{thm:compute_g})
    \State // Takes 1 pass, $O(n)$ space, $\tilde O(m(t+T_0))$ time.
	\State Store $\delta_y^{t+1}\gets\mH^{-1}\mA^\top\Phi''(x)^{-\frac{1}{2}} g$ via Laplacian solver on $\mH$, in $\tilde O(n)$ time, $\tilde O(n)$ space.
    \State Store $y^{t+1} \gets y^{t} + \delta_y^{t+1}$ in $O(n)$ space.
    \State Store $\delta_c^{t+1}\gets\mH^{-1}(\mA^\top x^{t}-b)$ by computing $(\mA^\top x^t-b)$ (using implicit representation of $x^{t}$), then calling Laplacian solver. Takes $1$ pass, $O(n)$ space, $\tilde O(m(t+T_0))$ time.
    \State Store $(t+1)$-transcript $\mathrm{Trans}^{t+1}=(\mathrm{Trans}^t, \mH, \delta_y, \delta_c, y, \mu, \otau, g)$. Increases size of transcript by $O(n)$.
	\State // For $\delta_x^{t+1}= \Phi''(x)^{-\frac{1}{2}}g-\omT^{-1}\Phi''(x)^{-1}\mA(\delta_y^{t+1}+\delta_c^{t+1})$, and
    $\delta_s^{t+1}= \mu\mA\delta_y^{t+1}$,
	\State // we have access to $x^{t+1}\gets x+\delta_x^{t+1}$ and $s^{t+1}\gets s+\delta_s^{t+1}$ via Lemmas \ref{cor:access_x} and \ref{cor:access_s}.
    \State \Return $(t+1)$-transcript $\mathrm{Trans}^{t+1}$
\EndProcedure
\end{algorithmic}
\end{algorithm}
\begin{lemma}\label{lem:ipm:periter}
    Given $t-1$-transcript of the IPM, we can compute w.h.p~the $t$-transcript in $\tilde O(1)$ passes, $\tilde O(m(t+T_0))$ time, and $\tilde O(n)$ additional space. %
\end{lemma}
\begin{proof}
    Our statement follows directly from the above theorems by adding the number of passes, space, and time complexity in each line of the algorithm.
\end{proof}

We can now prove the following:
\IPM*
\begin{proof}
    By \Cref{lem:ipm}, the algorithm runs for $O(\sqrt{n}\log(\mu^\init/\mu^\target))$ iterations.\\
    We use \Cref{lem:ipm:periter}. We have $\tilde O(1)$ passes per iteration, so $O(\sqrt{n}\log(\mu^\init/\mu^\target))$ total.\\
    We can upper bound $t$ by $\sqrt{n}\log(\mu^\target/\mu^\init)$ to get $\tilde O(m(\sqrt{n}\log(\mu^\init/\mu^\target)+T_0)$ per iteration, or $\tilde O(T_0 mn\log^2(\mu^\init/\mu^\target))$ total across all iterations.\\
    Finally, we also use $\tilde O(n)$ additional space per iteration, for a total of $\tilde O (n^{1.5} \log(\mu^\init/\mu^\target))$ space.\\
    By \Cref{prop:querytime}, the query time is $O(T_0 +\sqrt{n} \log(\mu^\target/\mu^\init))$.
\end{proof}

We can also extend this to the 2-party communication setting.
\begin{theorem}\label{thm:2partyIPM}
    Consider the 2-party model, where both parties have the same
    $\mu^\target \in \R_{>0}$ and $\epsilon$-centered (see \Cref{def:epscentered}) $(x^\init,s^\init,\mu^\init)$ in implicit representation (i.e., a party can query any entry $x^\init_e$ or $s^\init_e$, provided they know the corresponding edge $e$ and its cost and capacity). 
    Then, with 
    $\tilde O (n^{1.5} \log(\mu^\init/\mu^\target))$ bits of communication, both parties jointly construct the same $\epsilon$-centered $(x,s,\mu^\target)$ in implicit representation.
    The communication protocol is randomized and the output is correct w.h.p.
\end{theorem}

\begin{proof}
    We show the claim iteratively. Suppose both parties know the transcript up to iteration $t$. We claim that with only $\tilde O(n)$ bits of communication, they are able to jointly compute the transcript for iteration $t+1$. Combining this with \Cref{lem:ipm}, which states that the algorithm runs for $\tilde O(\sqrt{n})$ iterations, gives the theorem.

    To show the claim for a single iteration, note that both players can simulate the above steaming algorithm. Specifically, as we assume both parties have access to the $(t-1)$-transcript, they can jointly simulate a single pass over the input of the streaming algorithm while communicating only $\tilde O(n)$ space.
    First player A iterates over all their edges and performs exactly the calculations of the streaming algorithm. 
    By \Cref{lem:ipm:periter} a pass over the input uses only $\tilde O(n)$ additional space when having access to the $(t-1)$-transcript.
    So player B can continue simulating the current pass over the input after player A sends those additional $\tilde O(n)$ information.
    With this, the two parties jointly compute the new $t$-transcript, i.e., they now both have complete access to the $t$-transcript, giving us the theorem recursively.
    
\end{proof}

\section{Constructing the Gradient \texorpdfstring{$g$}{}}\label{sec:low:gradient_g}
\label{sec:gradient}

The main goal of this section is to prove the following theorem, which shows that we can construct exact gradient-like vectors with low space.

\thmGradient

The proof will make use the following two lemmas from \cite{bln+20}.

\begin{lemma}[{\cite[Corollary 7.4]{bln+20}}]\label{lem:corollary_7.4_in_bln+20}
Given $h,v\in\R^k$ we can compute in $O(k \log k)$ time
\begin{align*}
u = \arg\max_{\|v  w\|_2+\|w\|_\infty \le 1}
\langle h, w \rangle .
\end{align*}
\end{lemma}

Note that for $v = C \log(4m/n)\cdot \sqrt{\tau}$ (where $C$ is some constant) we have $\|vw\|_2 = C \log(4m/n) \|w\|_\tau$ and thus $\|v  w\|_2+\|w\|_\infty = \|w\|_{\tau+\infty}$ by \Cref{def:tauinftynorm}.
Thus we have
$$
u = \arg\max_{\|v  w\|_2+\|w\|_\infty \le 1}
\langle h, w \rangle = h^{\flat(\tau)}.
$$
Thus to get \Cref{thm:g_construction}, we will use $v=C \log(4m/n)\cdot \sqrt{\round_{1+\epsilon'}\tau}$ in \Cref{lem:corollary_7.4_in_bln+20,lem:lemma_7.5_in_bln+20}.

\begin{lemma}[{\cite[Lemma 7.5]{bln+20}}]\label{lem:lemma_7.5_in_bln+20}
Given $\oh, \ov, \of \in \R^k$, $h, v \in \R^n$, and a partition $\cup^k_{
i=1} I_i = [n]$, such that
\begin{itemize}
\item $h =\sum_{i=1}^k \oh_i \cdot {\bf 1}_{ I_i }$,
\item $ v = \sum_{i=1}^k \ov_i \cdot {\bf 1}_{ I_i }$ 
\item $\of_i = \sqrt{ |I_i| }$, for all $i \in [k]$ (This is the square root of the frequency)
\end{itemize}

Let $\ou \in \R^k$ be the maximizer of
\begin{align*}
\ou := \underset{
\ow \in \R^k ~:~ \| ( \of  \ov)  \ow \|_2+ \| \ow \|_{\infty} \leq 1}{\arg\max}
\langle \oh  \of  \of, \ow \rangle.
\end{align*}
We can define $u \in \R^n$ via $u = \sum_{i=1}^k \ou_i \cdot {\bf 1}_{ I_i } $, then
\begin{align*}
    \langle v, u \rangle = \max_{w \in \R^n ~:~ \| v  w \|_2 + \| w \|_{\infty} \leq 1} \langle h, w \rangle.
\end{align*}
\end{lemma}

\begin{proof}[Proof of \Cref{thm:g_construction}]    
    We round $\tau_e$ to nearest $(1+\epsilon')^i$ for some integer $i$, which we are guaranteed we can query for each edge we receive from the stream. 
    The rounding to $(1+\epsilon')^i$ discretizes the vector $\tilde \tau$ which will later allow us to project the vector into a smaller dimensional space.
    
    We use $\round_{1+\epsilon'}( \tau_e )$ to denote the number after rounding $\tau_e$.
    We have $n/m \le \tau \le 1$, so there are at most $O(\epsilon'^{-1}\log m)$ possible values. Formally speaking, let $I$ be defined as 
\begin{align*}
    I :=  \{ i \in \Z ~|~ n/m\leq (1+\epsilon')^i \leq 1 \} 
\end{align*}
Then, 
\begin{align*}
    |I| = O(\epsilon'^{-1} \log m)
\end{align*}

    We also round $v_e$ to nearest $j \cdot \epsilon'$ for some integer $j$. Let us use $\round_{\epsilon'}(v_e)$ to denote the number after rounding $v_e$.

    We have $-1\le v \le 1$ 
    so there are at most $O(\epsilon'^{-1})$ possible values. 
    Formally speaking, we define set $J$ as follows
    \begin{align*}
        J := \{i \in \Z ~|~ 0 \leq j \cdot \epsilon' \leq 2\}
    \end{align*}
    Then 
    \begin{align*}
        |J| = O(\epsilon'^{-1})
    \end{align*}

    Let $k$ be the number of possible pairs $(i,j) \in I \times J$. 
    
    Using the the above simple argument, we know the maximum number of groups is
    \begin{align*}
        k 
        = & ~ |I| \times |J| \\
        = & ~ O(\epsilon'^{-1} \log m) \cdot O(\epsilon'^{-1}) \\
        = & ~ O(\epsilon'^{-2}\log m) 
    \end{align*}

    Let $p : E\to[k]$ be the function that for any edge $e$ returns the group we rounded to.
    
    Throughout our pass of the stream, we count how often we round to each pair $(i,j)$, which takes $k=O(\epsilon'^{-2}\log m)$ space and $O(mT)$ time.
    
    In Lemma~\ref{lem:lemma_7.5_in_bln+20} let $\og\in\R^k$ be the vector that counts how often we rounded to each group.
    Let $\oh, \ov \in\R^k$ such that 
    \begin{align*}
        f(\round_{\epsilon'} (v_e)) = \oh_{p(e)}
    \end{align*}
    and 
    \begin{align*}
        C \log(4m/n)\cdot\sqrt{\round_{1+\epsilon'}( \tau_e )} = \ov_{p(e)}.
    \end{align*}

    Then by Lemma~\ref{lem:lemma_7.5_in_bln+20} and \Cref{def:tauinftynorm} the maximizer
    \begin{align*}
   \ou = \argmax_{\ow \in \R^k: \|\of  \ov  \ow \|_2+\|\ow\|_\infty \le 1} \langle\oh  \of,\ow\rangle
    \end{align*}
    has the property
    \begin{align*}
   \ou _{p(e)} = \left(\argmax_{w: \|w\|_{(\round_{1+\epsilon'}(\tau))+\infty}\le 1} \langle w, f(\round_{\epsilon'}(v)) \rangle\right)_e = \left((f(\round_{\epsilon'}(v))^{\flat(\round_{\epsilon'}(\tau))}\right)_e=: g_e
    \end{align*}
    We can use Lemma~\ref{lem:corollary_7.4_in_bln+20} to compute this $\ou$ in $O(k\log k)=\tilde O(1)$ time.

    \paragraph{Query}
    Now consider the case where we want to query $g_e$ for a given edge $e$.
    To do this, we compute $p(e)$, which requires computing $v_e$ and $\tau_e$ and rounding them.
    Then we return $\ou_{p(e)}$.
    In total, this takes $T+O(1)$ time.
    
\end{proof}

\section{Initial Flow and Complete Algorithm}\label{sec:initial_point:flow}

We now show how to initialize our interior point method and start from a feasible point near the central path. Then, we show that after running our algorithm we indeed have found an optimal feasible solution and can query the flow along any edge of it to $\epsilon$ accuracy.
This will conclude the proofs of \Cref{thm:main_stream,thm:main_communication}.

\subsection{Initial Point \& Auxiliary Graph}
\label{sec:initial_point:x0}
We have shown the following: %

\IPM*
In this section we discuss how to create an initial $(x^\init, s^\init, \mu^\init)$ that is feasible and on the central path and calculate $T_0$. The construction of \cite{bll+21} satisfies this and can be implemented in low space, so we can exactly use their construction.

\begin{tcolorbox}
\paragraph{Construction from \cite{bll+21}}
We assume that $u_e$ and $c_e$ are integral, and let $W$ be the maximum absolute value of $u_e$ and $c_e$ over all edges $e$.

Given a graph $G = (V, E)$, let $\tilde{G} = (V \cup \{z\}, \tilde{E})$ where $\tilde{E} = \{(v, z), (z, v) \mid v \in V\}$. That is, we add a bi-directional star rooted at $z$ into $G$. We let $\begin{bmatrix} x^\init \\ \tilde{x}^\init \end{bmatrix}$ denote an initial flow in $\tilde{G}$ where $x^\init$ and $\tilde{x}^\init$ specify the flow values on $E$ and $\tilde{E}$, respectively. For each $e \in E$, we set $x^\init_e := u_e / 2$. For all star-edges $e \in \tilde{E}$, we first set $\tilde{x}^\init_e = 1$ (just to prevent them from having flow value too close to zero). Then, we additionally route the excess at each vertex induced by $x^\init$ using the star-edges from $\tilde{E}$. More formally, for each $v \in V$, if the flow excess at $v$ is $[\mA^\top x^\init - F_{es,t}]_v > 0$, then we set $\tilde{x}^\init_{(v,z)} := 1 + [\mA^\top x^\init - F_{es,t}]_v$ and $\tilde{x}^\init_{(z,v)} = 1$. Otherwise, the flow deficit at $v$ is $[F_{es,t} - \mA^\top x^\init]_v \geq 0$ and we set $\tilde{x}^\init_{(z,v)} = 1 + [F_{es,t} - \mA^\top x^\init]_v$ and $\tilde{x}^\init_{(v,z)} = 1$. The capacity $u_e$ and cost $c_e$ of each original edge $e \in E$ stay the same. For each star-edge $e \in \tilde{E}$, we set the capacity $\tilde{u}_e = 2 \tilde{x}^\init_e$ and $\tilde{c}_e := 50m \| u \|_\infty \| c \|_\infty$. Consider the modified linear program for minimum cost flow on $\tilde{G}$:
\begin{equation}\label{eq:modFlowLP}
\min_{\substack{\tilde{\mA}^\top \begin{bmatrix} x \\ \tilde{x} \end{bmatrix} = F_{es,t}\\
0\leq x_e \leq u_e \, \forall e \in E\\
0 \leq \tilde{x}_e \leq \tilde{u}_e \, \forall e \in \tilde{E}}} c^\top x + \tilde{c}^\top \tilde{x},
\end{equation}
where $\tilde{\mA}$ is an incidence matrix of $\tilde{G}$. 

For a small technical reason, we need to further modify the linear program since an incidence matrix is degenerate and our path following algorithm only works on a full rank matrix. Remove a single column that corresponds to the vertex $z$.

\end{tcolorbox}

We first show that this construction works in our low-space setting as well, and analyze the query time of each edge.

\begin{lemma}\label{lem:init_flow}
    The initial flow in the construction from \cite{bll+21} can be stored in $O(n)$ space. Given the capacity along an edge, the flow of each edge can be queried in $O(1)$ time.
\end{lemma}
\begin{proof}
    To keep track of the flow on edges incident to the star vertex we only need to calculate and store all the other vertex demands once, which takes $O(n)$ space and $O(1)$ time per query.
    
    For all other edges, the initial flow values are simply the half of the edge capacities, so we do not need to store them, as we read in the edge capacity every time we see an edge.
\end{proof}

Now, it is clear that the construction itself is feasible, so we need to show that it is centered. \cite{bll+21} also proves centrality.

\begin{lemma}[{Initial point for the modified min cost flow, Lemma 7.5 of \cite{bll+21}}]\label{lem:initialPointModFlow}
    Let $\begin{bmatrix} s^\init \\ \tilde{s}^\init \end{bmatrix} = \begin{bmatrix} c \\ \tilde{c} \end{bmatrix}$ and $\mu^\init = 100 m^2 W^3 \epsilon^{-1}$. Then, the point $\left( \begin{bmatrix} x^\init \\ \tilde{x}^\init \end{bmatrix}, \begin{bmatrix} s^\init \\ \tilde{s}^\init \end{bmatrix}, \mu^\init \right)$ is $\epsilon$-centered w.r.t. the linear program in \eqref{eq:modFlowLP}.

\end{lemma}

Thus this is a suitable initial point for our algorithm. 

Finally, we need to show that this modified flow does indeed solve the original flow problem. 
Let $\text{OPT}(\tilde{G})$ be the optimal value of the above LP. 
Since the star-edge have cost $>m\|c\|_\infty$, any flow that goes along these edges, i.e., some $(v,z),(z,w)$, would be cheaper if it was instead routed along any other $vw$-path using only original edges. So the optimal solution does not use these edges, and thus the optimal flow on the original graph and the new graph is identical.
This assumes that it is possible to re-route the flow. If that is impossible, then that would be there is no path from $v$ to $w$ in the residual graph and thus the demand of the original problem cannot be satisfied.
\begin{fact}\label{fact:sameopt}
    $\OPT(\tilde{G})=\OPT(G).$
\end{fact}

\subsection{Final Point \& Existence of Nearby Exact Solution}
\newcommand{\stx}{\begin{bmatrix} x \\ \tilde{x} \end{bmatrix}}
\newcommand{\sts}{\begin{bmatrix} s \\ \tilde{s} \end{bmatrix}}
\newcommand{\stfix}{\begin{bmatrix} x^\final \\ \tilde{x}^\final \end{bmatrix}}
\newcommand{\stoptx}{\begin{bmatrix} x^\OPT \\ \tilde{x}^\OPT \end{bmatrix}}
\newcommand{\stfis}{\begin{bmatrix} s^\final \\ \tilde{s}^\final \end{bmatrix}}

We now need to show that after running our algorithm we have indeed found a near optimal solution.

We prove the following lemma.

\begin{lemma}\label{lem:lastPointApproxOpt}
    Given an $\epsilon$-centered point $\left( \begin{bmatrix} x \\ \tilde{x} \end{bmatrix}, \begin{bmatrix} s \\ \tilde{s}\end{bmatrix}, \mu \right)$ of \eqref{eq:modFlowLP} where $\mu<\epsilon/n^2$, there exists $x^\OPT$ which is an optimal flow for $G$ for which for $\tilde x^\OPT=0$ we have $\|\stx-\stoptx\|_\infty\lesssim \epsilon W$.
\end{lemma}
\begin{proof}
    Note from \Cref{fact:sameopt} we have $x^\OPT$ being an optimal flow for $G$ is equivalent to $\stoptx$ being an optimal flow for $\tilde G$.
    Then, this lemma follows directly from \Cref{lem:lastPointApproxFinal,lem:finalApproxOpt} when taking $\mu\le\frac{\epsilon}{n^2}$.
\end{proof}

\begin{lemma}[{Extension of Lemma 7.7 of \cite{bll+21}}]\label{lem:lastPointApproxFinal}
    Given an $\epsilon$-centered point 
    $\left( \begin{bmatrix} x \\ \tilde{x} \end{bmatrix},
    \begin{bmatrix} s \\ \tilde{s}\end{bmatrix}, \mu \right)$ 
    of \eqref{eq:modFlowLP} where $\mu<1/n$, there exists $\begin{bmatrix} x^\final \\ \tilde{x}^\final \end{bmatrix}$ which is a feasible flow for $\tilde{G}$ where $c^\top x^\final+\tilde{c}^\top\tilde{x}^\final\le\OPT(\tilde{G})+\mu n$. 
    Moreover, $\|\stx-\stfix\|_\infty\lesssim\epsilon W$.
\end{lemma}

\begin{proof}
    The first part of the lemma is simply Lemma 7.7 of \cite{bll+21}.\\
    To show the moreover part, we use the following claim:
    \begin{claim}[{Claim A.6 of \cite{bll+21}}]\label{clm:finalclose}
        \[\left\|\Phi''\left(\stx\right)^{\frac{1}{2}}\left(\stfix-\stx\right)\right\|_{\infty}\lesssim\epsilon 
        \]
    \end{claim}
    Then we have for any edge $e$ (including star edges) $\phi''(x_e)=\left(\frac{1}{x_e^2}-\frac{1}{(u_e-x_e)^2}\right)\lesssim\epsilon$. Recall $0\le x_e\le u_e\le W$. Then 
    \[
        |x_e^\final-x_e|\lesssim\frac{\epsilon}{(1/x_e^2+1/(u_e-x_e)^2)^{1/2}}\le\frac{\epsilon}{(2/W^2)^{1/2}}\le \epsilon W.
    \]
    where the second step follows from $1/x_e^2 + 1/(u_e-x_e)^2 \geq 2/ W^2$ for all $0 \leq x_e \leq u_w \leq W$ (see Fact~\ref{fac:x_y_a}) and the last step follows from simple algebra.
\end{proof}

Now we need to show that this final point is also close to the exact optimal point to conclude that the point found by the IPM is too. Note that it suffices to show that $\stfix$ is close to an optimal point of $\tilde G$ as any optimal flow for $\tilde G$ would not use any of the new star edges.

\begin{lemma}\label{lem:finalApproxOpt}
    Let $\stfix$ be a feasible flow for $\tilde G$ where $c^\top x^\final+\tilde c^\top\tilde x^\final\le\OPT(\tilde G)+\mu n$. Then there exists an optimal flow $\stoptx$ of $\tilde G$ with $\tilde x^\OPT=0$ such that $\|\stfix-\stoptx\|_\infty\le\mu n^2$.
\end{lemma}
\begin{proof}
    $\stfix$ is feasible but not optimal. Recall that we can get an optimal flow by repeatedly improving the value of a feasible flow by finding a negative cost cycle in the residual graph induced by this flow, then augmenting the flow by this cycle. Let $\stoptx$ be the optimal flow created by augmenting $\stfix$ in this way.\\
    Note that our edge costs are integral; thus, each cycle has cost at most -1. Since we have $c^\top x^\final+\tilde c^\top\tilde x^\final\le\OPT(\tilde G)+\mu n$, and each cycle can only use $n$ edges, we have $\|\stfix-\stoptx\|_\infty\le\|\stfix-\stoptx\|_1\le\mu n^2$.
\end{proof}

\subsection{Putting Everything Together}
\label{sec:everything_together}

In this section, we prove the main result \Cref{thm:main_communication,thm:main_stream}.

\mainStreamFlows*

\begin{proof}
We use the interior point method \Cref{alg:short_step}, specifically \Cref{thm:IPM}.
This method requires an $\epsilon$-centered initial point $(x,s,\mu)$ with query time $T_0$, and a target $\mu^\target$.

We use \Cref{lem:init_flow} to find an $\epsilon/W$-centered initial point,
which also gives us a $T_0=\tilde O(1)$ initial query time. 
By \Cref{lem:initialPointModFlow} this constructed $\epsilon/W$-centered point uses $\mu^\init=100m^2W^4\epsilon^{-1}$.

We now run our interior point method on this initial point with target parameter $\mu^\target=\epsilon/(Wn^2)$.
\Cref{thm:IPM} tells us we need $\tilde O(n^{1.5}\log(W/\epsilon)$ space and $\tilde O(\sqrt{n}\log(W/\epsilon))$ passes to construct the resulting flow $x$.
The total runtime is $\tilde O(mn\log^2(W/\epsilon))$.

The constructed flow $x\in\R^E$ is not an optimal min-cost flow, since (i) the cost is slightly too large, (ii) the flow might use the star-edge that were added in the initial point construction.

However, by \Cref{lem:lastPointApproxOpt} and accuracy parameter $\epsilon/W$ we see that for $\mu^\target=\epsilon/(Wn^2)$ our constructed $x$ is such that there exists an optimal solution $x^\OPT$ with $\|x-x^\OPT\|_\infty\le\epsilon$. 
In particular, this optimal flow does not route any flow on the star-edges that were added in the initial point construct, so $x^\OPT$ is an optimal flow for the original input graph.

Thus $f=x^\OPT$ is the optimal flow, and $\of=x$ is the flow for which we can query edges, and we have $\|f-\of\|_\infty\le\epsilon$.

By \Cref{thm:IPM}, we use $\tilde O(n^{1.5}\log (W/\epsilon))$ space with query time on each edge $\tilde O(\sqrt{n}\log(W/\epsilon))$.
\end{proof}

Next, we analyze our algorithm in the context of communication complexity. In this setting, we are only limited by the number of bits sent between the players and not the bits stored, so we can use the following isolation lemma to find an exact flow.
\begin{lemma}[{\cite{ds08}, or Lemma 7.8 of \cite{bll+21}}]\label{lem:isolation}
    Let $\Pi = (G, b, c)$ be an instance for minimum-cost flow problem where $G$ is a directed graph with $m$ edges, the demand vector $b\in \{-W,\dots, W\}^V$, the cost vector $c\in \{-W,\dots , W \}^E$ and the capacity vector $u \in \{0, \dots, W \}^E$.
    
    Let the perturbed instance $\Pi' = (G, b, c' )$ be such that $c'_e = c_e + z_e$ where $z_e$ is a random number from the set $\left\{\frac{1}{4m^2W^2},\dots,\frac{2mW}{4m^2W^2}\right\}$. Let $x'$ be a feasible flow for $\Pi'$ whose cost is at most $\OPT(\Pi') + \frac1{12m^2 W^3}$ where $\OPT(\Pi')$ is the optimal cost for problem $\Pi'$. With probability at least 1/2, there exists an optimal feasible and integral flow $x$ for $\Pi$ such that $\|x - x'\|_\infty\le 1/3$.
\end{lemma}

\mainCommunication*
\begin{proof}
    \Cref{thm:main_communication} follows directly in the same way as \Cref{thm:main_stream} while applying \Cref{thm:2partyIPM}.
    For this, we need to show that we can indeed start with $(x^\init, s^\init, \mu^\init)$ for both parties.
    $\mu^\init$ is a single predetermined scalar, and
    for all non-star edges $e$, each party can query $x_e^\init$ and $s_e^\init$ for their own edge set, because they are simply $u_e/2$ and $c_e$, respectively, which each party knows for their own edge set.
    The flow on star edges $e'$ depend on the vertex demands: since there are only $O(n)$ vertices each party can add up the demands from their edges and send them to each other using $O(n)$ bits of communication to both know the total vertex demands and determine $x_{e'}^\init$ and $s_{e'}^\init$. 

    Then, we can continue with the algorithm (\Cref{thm:2partyIPM}) until both parties know the final transcript. Then each party can query the flow on their own edge set via the transcript.

    Now, to make the flow unique, we very slightly modify the above algorithm. Before we do anything we first perturb the edge costs $c$ to $c'$ according to \Cref{lem:isolation}. %
    Storing the perturbations takes $O(m)$ space, which can be larger than $O(n^{1.5})$, but crucially neither player needs to send this information to the other, as each keeps track of its own edges. 

    Then, after completing the algorithm, consider our final flow $\stfix$. 
    \Cref{lem:isolation} guarantees with probability at least $1/2$ that there exists an integral min cost flow $\stoptx$ with each edge value differing by at most $1/3$. Thus, with probability at least $1/2$, if whenever we query the flow along an edge we instead round it to the nearest integer, we recover $\stoptx$. Moreover, by \Cref{lem:finalApproxOpt}, we know $\|\tilde x^\final\|_\infty$ can be made less than $1/3$, so indeed the rounded $\stfix$ is feasible and optimal for $G$.

    Finally, we can boost our success probability to high probability by resampling the random edge costs and repeating our algorithm $O(\log n)$ times. We can verify whether each flow is feasible after rounding by comparing the vertex demands, i.e., having both parties exchange the net-flow on each vertex and checking that they add up to demand $b$, which uses $O(n)$ space. Then in the end, we choose the flow with the smallest cost among all constructed feasible flows. 
    
\end{proof}

\paragraph{Remark on Bit-complexity}
The bit complexity can be bounded via the following lemma from \cite{bll+21}.

\begin{lemma}[{\cite[Lemma 4.46]{bll+21}}]
    Let $W$ be the ratio of largest to smallest entry of $\phi''(x^\init)^{1/2}$ and let $W'$ be the ratio of largest to smallest entry of $\phi''(x)^{1/2}$ encountered throughout the algorithm.
    Then
    $$
    \log W' = \tilde O(\log W + \log(1/\mu^\final) + \log \|c\|_\infty).
    $$
\end{lemma}

Observe that by $\phi''(x) = 1/x_e^2 - 1/(u_e-x_e)^2$ and $\mu^\final = 1/(\epsilon \cdot \poly(n\max(\|c\|_\infty,\|u\|_\infty,\|b\|_\infty)))$, this lemma bounds the smallest value we might obtain for any $x_e$.
If we run the algorithm with some $\tilde O(\log (W/\epsilon))$ bit-length numbers, then we can guarantee the additive error is small multiplicative error, i.e., we always have a multiplicative approximation of $x$ and $\phi''(x)$.
Since the algorithm supports spectral approximation of $\mA^\top \mT^{-1}\Phi''(x)^{-1}\mA$, and entry-wise approximation of $\tau$ and $x$ when defining $g$ (\Cref{sec:gradient}), the algorithm and correctness is already accounting for such errors.
\cite{bll+21} have proven that defining the entire algorithm using some element-wise approximation $\ox,\os,\otau$ still guarantees convergence, where this was exploited to develop data structures that maintain/compute only $\ox,\os,\otau$ rather than exact $x,s,\tau$.

\section*{Acknowledgment}

Jan van den Brand is supported
by NSF Award CCF-2338816 and CCF-2504994. 
Albert Weng is supported by NSF Award CCF-2338816.

\bibliographystyle{alpha}
\bibliography{ref}

\appendix

\section{Discussion on Space Upper and Lower Bounds}\label{sec:lower_bound}

In this section we state the lower bound of \cite{ChakrabartiJWY25} for $\ell_p$-regression and restate it for undirected maximum flow (\Cref{thm:reduction_from_regression_to_flow}).
We also discuss that on unweighted (i.e., capacity 1) graphs, the max-flow can easily be stored in subquadratic space because it uses at most $O(n^{1.5})$ edges (\Cref{thm:unweightedflow}).

\begin{theorem}[{\cite[Theorem 5]{ChakrabartiJWY25}}]\label{thm:regressionlowerbound}
    Let $p\in(1,\infty]$ and let $q=p/(p-1)$ be its H\"older conjugate exponent.
    Let $\epsilon\in(0,1/(8q))$ and $d\in\N$.
    Any randomized algorithm that computes a $B$-bit summary of $a\in\{\pm1\}^d$ from which $\ox\in\R^d$ can be produced such that, with probability at least 2/3, we have $a^\top \ox = d$ and $\|\ox\|_p \le (1+\epsilon)\min_{a^\top x = d}\|x\|_p$ requires $B=\Omega(d)$.
\end{theorem}

Now, we're ready to prove \Cref{thm:reduction_from_regression_to_flow}.
\begin{theorem}\label{thm:reduction_from_regression_to_flow}
    Any randomized algorithm $\cE$ that computes $\cE(G)\in\{0,1\}^B$ such that with probability at least 2/3, another algorithm $\cD(\cE(G))$ returns a $(1+\epsilon)$-approximate $st$-max flow for $\epsilon\in(0,1/36)$, requires $B=\Omega(n^2)$.
\end{theorem}

\begin{proof}
    Given $\cE$ we will construct a randomized algorithm that receives $a\in\{\pm1\}^d$ for $d=n^2$, and computes a $B=O(B)$ bit summary.

    Then we show that from this summary, we can produce an $\ox \in\R^d$ such that $a^\top \ox = d$ and $\|\ox\|_\infty < (1+1/8)\min_{a^\top x} \|x\|_\infty$.

    Hence by \Cref{thm:regressionlowerbound} ($p=\infty, q=1$), this implies $B=\Omega(n^2)$.

    \paragraph{Constructing the $O(B)$-bit summary}
    Given the vector $a\in\{\pm1\}^d$, we consider a complete bipartite graph $G$ with edges oriented from left to right where the capacity of the $i$-th edge is $2+a_i \in \{1,3\}$ forall $i \in [d]$.
    We add an extra vertex $s$ that connects to every vertex on the left, and an extra vertex $t$ reachable from every vertex on the right. The capacity of these edges is $4d$ so that the maximum flow will use $f_e=c_e$ for all edges of the bipartite graph.
    We run $\cE(G)$ to get the $B$-bit summary $s \in \{0,1\}^B$ of a $(1+\epsilon)$-approximate $st$-maximum flow.
    We also run the decoding algorithm $\cD(s)$ on that summary to obtain the approximate flow $f\in\R^E$ and compute $\sum_{i=1}^d a_i (c_i - \max(1,f_i))$. Attach this sum to our summary $s$.
    Thus we now have a $(B+\tilde{O}(1))$-bit summary.

    \paragraph{Constructing vector $\ox$ from $O(B)$-bit summary.}
    Given the summary, we can construct the $(1+\epsilon)$-approximate max flow $f$ and know 
    \begin{align}\label{eq:def_r}
    r := \sum_{i=1}^d a_i (c_i - \max(1,f_i)).
    \end{align}
    We define for all $i \in [d]$
    \begin{align}\label{eq:def_ox_i}
        \ox_i := (\max(1,f_i)-2) \cdot \frac{1}{1-r/d}.
    \end{align}
    Observe that for the true max-flow, each edge of the bipartite graph would be at full capacity. So if the flow was exact, we would have $f_i-2 = a_i$ and $r=0$. In particular, we would trivially have 
    \begin{align*}
        a^\top \ox = \sum_{i=1}^d a_i^2 = d.
    \end{align*}
    \noindent
    For our approximate flow, we still satisfy that property because
    \begin{align*}
        \sum_{i=1}^d a_i \ox_i = & ~ \sum_{i=1}^d a_i (\max(1,f_i)-2) \cdot \frac{1}{1-r/d} \\
        = & ~ \sum_{i=1}^d a_i \cdot (\max(1,f_i)-c_i+a_i) \cdot \frac{1}{1-r/d} \\
        = & ~ ( (\sum_{i=1}^d a_i^2 ) + \sum_{i=1}^d a_i\cdot(\max(1,f_i)-c_i) ) \cdot \frac{1}{1-r/d} \\
        = & ~ ( d+ (\sum_{i=1}^d a_i \cdot (\max(1,f_i)-c_i))) \cdot \frac{1}{1-r/d} \\
        = & ~ (d-r) \cdot \frac{1}{1-r/d} \\
        = & ~ d
    \end{align*}
where the first step follows from Definition of $\ox_i$ (see Eq.~\eqref{eq:def_ox_i}), the second step follows from $c_i = a_i+2$, the third step follows from separating two summations, the forth step follows from $\sum_{i=1}^d a_i^2 = 1$, the fifth step follows from Definition of $r$ (see Eq.~\eqref{eq:def_r}), and the last step follows from simple algebra.
    
    Lastly, this $\ox \in \R^d$ is also an approximation of 
    \begin{align*}
    \min_{a^\top x=d} \|x\|_\infty,
    \end{align*}
    because the optimal solution is $x=a$ which has $\|x\|_\infty=1$ and for our $\ox$ we have
\begin{align*}
    \| x \|_{\infty} \le & \frac{1}{1-r/d}
\end{align*}
where for the true max-flow $f^*$ of size $F\le 4d$ (by $c_i \le 4$ for all $i \in [d]$) we have
\begin{align}\label{eq:upper_bound_r}
r 
= & ~ \sum_{i=1}^d a_i (c_i - \max(1,f_i)) \notag \\
= & ~ \sum_{i=1}^d a_i (c_i - f^*) + \sum_{i=1}^d a_i (f^*_i -\max(1,f_i)) \notag \\
= & ~ \sum_{i=1}^d a_i (f^*_i -\max(1,f_i)) \notag \\
\le & ~ \sum_{i=1}^d (f^*_i -\max(1,f_i)) \notag \\
\le & ~ \epsilon F \notag \\
\le & ~  4 \epsilon d
\end{align}
where the first step follows from Eq.~\eqref{eq:def_r}, the last step follows from $F \leq 4d$, the second step follows from simple algebra, the third step follows from $c_i = f_i^*$ for all $i \in [d]$, the forth step follows each summation term is positive and $a_i \in \{\pm1\}$, the fifth step follows from property of $f_i$, and the last step follows from $F\leq 4d$.

Thus
\begin{align*}
\frac{1}{1-r/d}
\le
\frac{1}{1-4\epsilon}
\leq 1+1/8
\end{align*}
where the first step follows from Eq.~\eqref {eq:def_r}, the last step follows from $\epsilon \in (0, 1/36)$.
\end{proof}

Our work focuses on the weighted case. For unweighted graphs with capacity 1 on every edge, a max flow (assuming no unnecessary cycles) always uses at most $O(n^{1.5})$ edges. Thus they can be trivially stored in that much space.
It's interesting to see that now the same upper bound holds for the weighted case, despite graph with arbitrary capacities using up to $O(n^2)$ edges.

\begin{theorem}\label{thm:unweightedflow}
    Given a graph $G=(V,E)$ with $|V| =n$ and unit-capacity, then any size $F$ $st$-max flow (after removing cycles) uses at most $2 n \sqrt{F} $ edges. %

    In particular, since $F \leq n$, then we get $2n^{1.5}$ edges.
\end{theorem}

\begin{proof}
After removing cycles from the max-flow, the edges used by the flow form a directed acyclic graph (DAG). From now, we consider the graph that consists of only the flow-edges. Since that is a DAG, consider the topological order.

Let $g$ be some parameter. 
We split the sequence of ordered vertices into $g$ groups of $n / g$ vertices. Let $V_1, V_2, \cdots, V_{g}$ denote those $n$ groups and each $|V_i| = n / g$ for all $i \in [g]$.

If we consider all the max-flow edges, there are now 2 types: 
\begin{itemize}
    \item Type I: edges within a single group, 
    \item Type II: edges that go from one group to another group.
\end{itemize}

First we will explain how to compute the number of Type I edges. 
Within a single group there are at most $(n/g)^2$ edges, and since there are $g$ groups, we have 
\begin{align*}
(n/g)^2 \cdot g = n^2/g %
\end{align*}
edges of that Type I in total.

For any set $W \subset V$, we use the $\delta(W)$ to denote the size of cut between $W$ and $V \setminus W$.

Now, let us consider edges of type II that go across different groups.
For any such edge $e$, assume it leaves some group $V_k$, then the edge is cut by the cut $W_k \subset V$ where
\begin{align*}
W_k := \bigcup_{i=1}^k V_i, \mathrm{~and~} V\setminus W_k = \bigcup_{i=k+1}^{g} V_i,
\end{align*}

\noindent
because all flow-edges are oriented left to right by topological order assumption.
Thus the number of type II edges is bounded by the sum of sizes of all such cuts, i.e., 
\begin{align*}
\sum_{k=1}^{g}\delta(W_k).
\end{align*}

\noindent
Since the graph is a DAG and $W_k$ contains the first $kn/g$ vertices of the topological order, the cut has only outgoing edges, in particular, $\delta(W_k)= F$ the size of the flow. Thus 
\begin{align*}
\sum_{k=1}^{g}\delta(W_k) = g \cdot F %
\end{align*}
So there are at most $gF$ many edges of type II.

Overall, if count two types of edges together, there are
\begin{align*}
n^2/ g + g \cdot F.
\end{align*}
To minimize the above quantity, we choose $g = n/\sqrt{F}$. Thus, we obtain
\begin{align*}
2 n \sqrt{F}
\end{align*}
By unit-capacity we have $F\le n$, implying at most $2n^{1.5}$ edges in the max-flow.

\end{proof}

\section{Basic Algebra Facts and Inequalities}

\begin{fact}\label{fac:x_y_a}
Let $0 \leq x \leq y \leq a$. Then we can show
\begin{align*}
\frac{1}{x^2} + \frac{1}{(y-x)^2} \geq \frac{2}{a^2}
\end{align*}
\end{fact}
\begin{proof}
Since $0 < y-x \leq a - x$, then $(y-x)^2 \leq (a-x)^2$. This implies that 
\begin{align}\label{eq:x_y_a_1}
\frac{1}{(y-x)^2} \geq \frac{1}{(a-x)^2}
\end{align}

Let $f(x) := \frac{1}{x^2}$. Note that $f(x)$ is a convex function in $[0,a]$, thus
\begin{align}\label{eq:x_y_a_2}
f(x) + f(a-x) \geq 2 f( a-x + x) = \frac{2}{a^2}
\end{align}

Thus, we have 
\begin{align*}
\frac{1}{x^2} + \frac{1}{(y-x)^2} \geq & ~ \frac{1}{x^2} + \frac{1}{ (a-x)^2 } \\
\geq & ~ \frac{2}{a^2}
\end{align*}
where the first step follows from Eq.~\eqref{eq:x_y_a_1}, and the second step follows from Eq.~\eqref{eq:x_y_a_2}. 

Thus, we complete the proof.
\end{proof}

\end{document}